\documentclass[a4paper,11pt]{article}

\usepackage[utf8]{inputenc}
\usepackage[T1]{fontenc}
\usepackage[english]{babel}
\usepackage{geometry}
\usepackage{amsmath,amssymb,amsfonts,amsthm,bm}
\usepackage{graphicx}
\usepackage{placeins}
\usepackage{xcolor}
\definecolor{linkblue}{RGB}{0,0,160}
\usepackage[colorlinks=true, linkcolor=linkblue, citecolor=linkblue, urlcolor=linkblue]{hyperref}
\usepackage{cleveref}
\usepackage{siunitx}
\usepackage[font=small, labelfont=bf]{caption}
\usepackage{booktabs}
\usepackage{tabularx}
\usepackage{array}
\usepackage{microtype}
\usepackage{algorithm}
\usepackage{algpseudocode}
\usepackage{orcidlink}
\usepackage{csquotes}
\usepackage{subcaption}
\usepackage{enumitem}
\usepackage{float}
\usepackage[en-GB]{datetime2}
\usepackage[
  style=numeric,
  sorting=none,
  maxcitenames=2,
  date=year,
  doi=true,
  isbn=false,
  url=false,
  eprint=false,
  giveninits=true
]{biblatex}

\DeclareNameAlias{author}{family-given}
\DeclareFieldFormat[article]{title}{#1}
\DeclareFieldFormat[article]{journaltitle}{\textit{#1}}
\DeclareFieldFormat[article]{volume}{\textbf{#1}}
\DeclareFieldFormat[article]{number}{(#1)}
\DeclareFieldFormat[article]{pages}{#1}
\renewbibmacro{in:}{}

\newtheorem{proposition}{Proposition}[section]
\newtheorem{definition}[proposition]{Definition}

\crefname{figure}{Fig.}{Figs.}
\crefname{table}{Tab.}{Tabs.}
\crefname{equation}{Eq.}{Eqs.}
\crefname{definition}{Definition}{Definitions}
\crefname{proposition}{Proposition}{Propositions}

\begin{document}

\title{Local Information Operators for Spatial Identifiability\\
in Distributed-Parameter Inverse Problems\\
in Computational Mechanics}

\author{
Tammam Bakeer\orcidlink{0000-0001-6475-3907}
\thanks{Independent Researcher, Dresden, Germany. Corresponding author. Email: \href{mailto:mail@bakeer.de}{mail@bakeer.de}.}
}

\date{13 May 2026}
\maketitle

\begin{abstract}
In distributed-parameter inverse problems in computational mechanics, spatially varying fields are inferred from noisy, indirect, and heterogeneous observations. The relevant identifiability question concerns which spatial perturbation patterns of the field are distinguishable under a specified sensing and excitation programme. This paper develops a local information-operator framework for this purpose. Around a nominal parameter field, the parameter-to-observation map is linearized and the likelihood contribution to posterior precision is interpreted as an operator on parameter-field perturbations. For locally linearized Gaussian models with parameter-independent covariance, this operator is equivalently Fisher information, Gauss-Newton data-misfit curvature, and a noise-weighted sensitivity Gramian.

The framework separates pointwise visibility from spatial identifiability. The diagonal gives a coordinate-dependent local information density, while the full kernel and metric- or prior-preconditioned spectra rank spatial patterns that are strongly visible, weakly visible, or locally invisible. Heterogeneous observation blocks are assembled in a common parameter space; information is additive only under conditional independence, whereas correlated errors require the full joint covariance. Model discrepancy, nuisance parameters, and prior information modify the same geometry through covariance inflation, Schur-complement information loss, and prior-preconditioned modes. Examples cover analytic beam kernels, two-span support coupling, static-dynamic fusion for flexural-rigidity identification, and two-dimensional damage-field reconstruction in a leading information subspace. The operator view supports interpretation of identifiability, sensor complementarity, and reduced reconstruction in distributed-parameter inverse problems.
\end{abstract}


\newpage
\section{Introduction}
\label{sec:introduction}

Inverse problems in computational mechanics increasingly require the inference of
spatially distributed parameter fields from indirect, noisy, and heterogeneous
observations. These fields may describe stiffness or compliance, damage, corrosion,
section loss, fracture resistance, boundary or support degradation, prestress loss,
scour, uncertain loading, or environmental exposure. They are central to reliability
assessment, load rating, remaining-life estimation, digital-twin updating, and
maintenance decision-making. In such problems, however, the essential question is not
only which field fits the data, but also which perturbation patterns of that field are
distinguishable by the available observations. A sensing programme may strongly
constrain some spatial patterns, weakly constrain others, and leave some directions
effectively invisible. \emph{Spatial identifiability} therefore depends on the combined action
of the mechanics, excitation, sensor layout, observation type, parameterization, prior
information, and noise model.

Recent structural and mechanics-oriented studies demonstrate the practical relevance of
distributed-parameter inference, including distributed stiffness updating, bridge
flexural-rigidity estimation, digital-twin updating, and sensor-informed structural
identification~\cite{adhikariDistributedParameterModel2010,lohnerHighFidelityDigitalTwins2023,
airaudoAdjointbasedDeterminationWeaknesses2023,youDistributedBendingStiffness2023,
bakeerSensorInformativenessIdentifiability2025a,diazMergingExperimentalDesign2023}.
They also show that the inferred fields can depend strongly on sensing geometry,
excitation richness, parameter-field representation, and uncertainty modelling.
Classical inverse-problem theory explains why this dependence is not accidental:
ill-posedness is intrinsic to indirect field inference rather than
exceptional~\cite{beckParameterEstimationEngineering1977,tarantolaInverseProblemTheory2005,stuartInverseProblemsBayesian2010}.
This situation is illustrated conceptually in
\cref{fig:ill_posedness_landscape}: the data misfit may be small along an extended
valley, so several parameter fields can fit the observations comparably well while
differing along weakly constrained directions.

\begin{figure}[!htbp]
\centering
\includegraphics[width=0.7\textwidth]{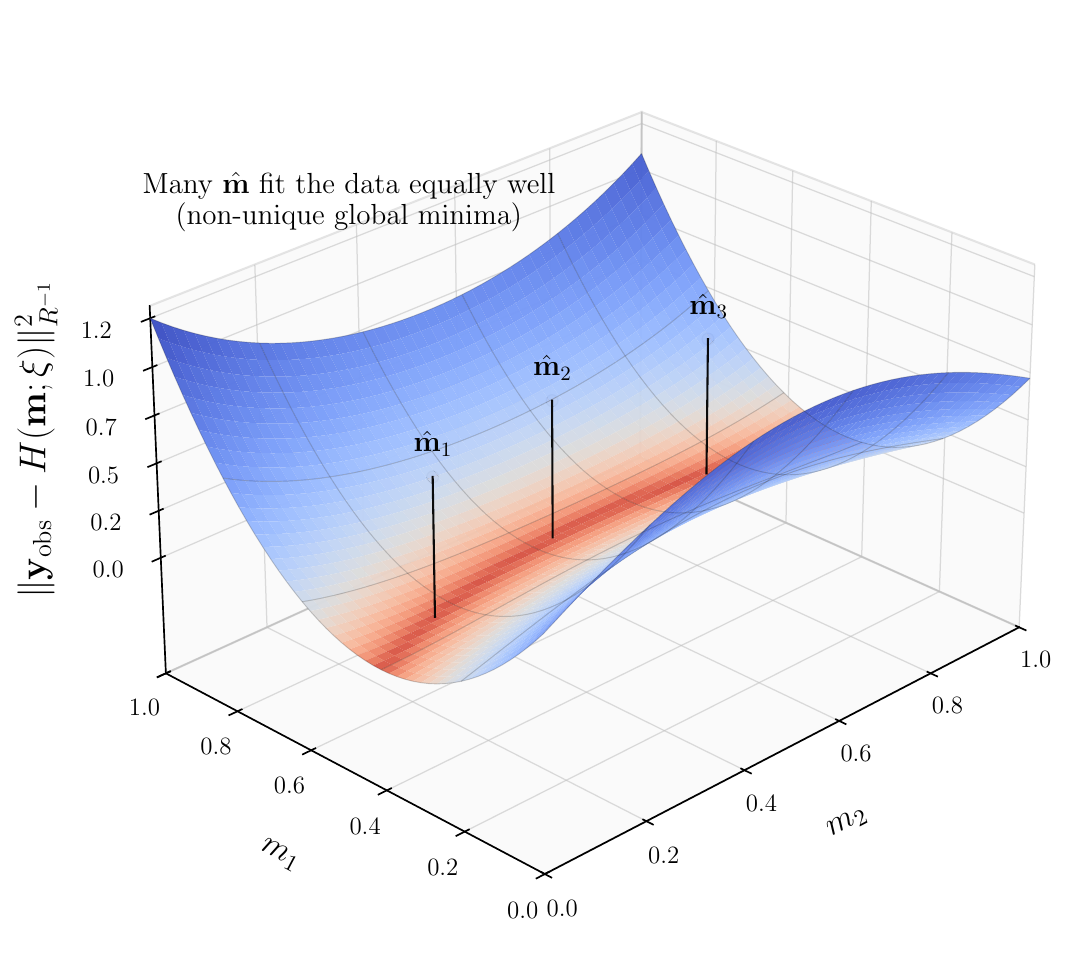}
\caption{Conceptual loss landscape for an ill-posed inverse problem. The valley of small weighted residual norm contains multiple estimates \(\hat m_1,\hat m_2,\hat m_3\) with comparable data fit, illustrating that the observations may constrain only some parameter combinations while leaving other directions weakly constrained or nearly non-identifiable.}
\label{fig:ill_posedness_landscape}
\end{figure}

Several established lines of work address parts of this problem. In structural dynamics
and model updating, information-theoretic and Fisher-information-based sensor-placement
methods use entropy, determinant-type criteria, and related scalar utilities to rank
sensor layouts or measurement
channels~\cite{kammerSensorPlacementOnorbit1991,sanayeiSensorPlacementParameter2002,
papadimitriouEntropyBasedOptimalSensor2000,papadimitriouOptimalSensorPlacement2004,
yiMethodologyDevelopmentsSensor2012,barthorpeEmergingTrendsOptimal2020,
ebrahimianInformationTheoreticApproachIdentifiability2019,ercanBayesianOptimalSensor2023,
antilFisherInformationBasedSensorPlacement2026}.
Distributed-parameter systems have likewise used sensitivity and Fisher-information
ideas for experiment and sensor-location
selection~\cite{vandewouwerApproachSelectionOptimal2000}. In parallel, Bayesian optimal
experimental design has developed expected-information-gain and posterior-uncertainty
criteria, together with scalable algorithms for partial-differential-equation
(PDE)-constrained inverse problems~\cite{huanSimulationbasedOptimalBayesian2013,
alexanderianAOptimalDesignExperiments2014,alexanderianOptimalExperimentalDesign2021,
ryanReviewModernComputational2016,huanOptimalExperimentalDesign2024}. Large-scale Bayesian inverse problems have also
shown that the data-informed geometry is often effectively low-dimensional, as expressed
by low-rank misfit Hessians, likelihood-informed subspaces, and low-rank posterior
corrections~\cite{flathFastAlgorithmsBayesian2011,
bui-thanhComputationalFrameworkInfiniteDimensional2013,
cuiLikelihoodinformedDimensionReduction2014,spantiniOptimalLowrankApproximations2015,
scheffelsLikelihoodinformedModelReduction2025}.
Systems and estimation theory provide a related Gramian language for observability and
estimability~\cite{moorePrincipalComponentAnalysis1981,scherpenBalancingNonlinearSystems1993,
lallEmpiricalModelReduction1999,himpeEmgrEmpiricalGramian2018,
jauffretObservabilityFisherInformation2007a,lewisRoleObservabilityGramian2022,
kunwooObservabilityGramianBayesian2023}.

These developments are commonly used to approximate posteriors, optimize scalar design
criteria, assess practical identifiability, or identify data-informed dimensions in
abstract parameter spaces. In mechanics applications, however, the spatial diagnostic
interpretation of the full local information operator remains less explicit. Existing
scalar design criteria and diagonal information measures are useful, but they do not by
themselves reveal which stiffness, compliance, damage, boundary, or load-field patterns
are visible, weakly visible, redundant, complementary, or confounded under a specified
observation programme. Scalar criteria such as trace, determinant, entropy, or minimum
eigenvalue compress the information geometry into a single number. They therefore do not
show whether a new observation fills an existing information gap, reinforces already
informed patterns, or leaves important weak directions unchanged.

The main viewpoint of this paper is that local spatial identifiability in
distributed-parameter mechanics is an \emph{operator-level property}. Around a nominal
parameter field \(m_0\), the parameter-to-observation map is linearized, and the
likelihood contribution to posterior precision defines an information-weighted
bilinear form on parameter-field perturbations. For the locally linearized Gaussian
model with parameter-independent covariance, this quadratic form can be read as local
Fisher information, Gauss--Newton data-misfit curvature, or a noise-weighted
parameter-output sensitivity Gramian. The paper uses this equivalence to develop a
mechanics-oriented spatial diagnostic for distributed-parameter identifiability.
The term observability is used here in a parameter-identification sense: it refers to
the local distinguishability of parameter-field perturbations through the
parameter-to-observation map, not to classical state observability in control theory.

The emphasis is placed on the full operator rather than only on scalar criteria or
diagonal entries. Its diagonal gives a pointwise visibility density, whereas its kernel,
metric-dependent spectrum, and prior-preconditioned modes rank and couple the spatial
patterns that are locally distinguishable by the chosen sensing and excitation
programme. In this sense, the paper organizes Fisher-information, Bayesian optimal
experimental design (OED), likelihood-informed subspace (LIS), and Gramian ideas into a
mechanics-specific operator language for spatial identifiability.

The framework is local to \(m_0\), to the assumed covariance model, and to the chosen
parameter metric or prior. It is independent of the particular mechanics model and
observation type, provided the local Jacobian and covariance model are specified.
Point sensors, distributed strain, full-field measurements, modal data, time histories,
and frequency-response observations enter through the same parameter-space geometry.
Independent observation blocks add at the operator level; correlated errors require the
full joint covariance. Parameterization, prior metric, model discrepancy, nuisance
parameters, and correlated noise modify this local geometry.
\Cref{fig:general_setup} summarizes the representative mechanics setting underlying the
distributed-parameter inverse problems considered here.

\begin{figure}[!htbp]
\centering
\includegraphics[width=0.72\textwidth]{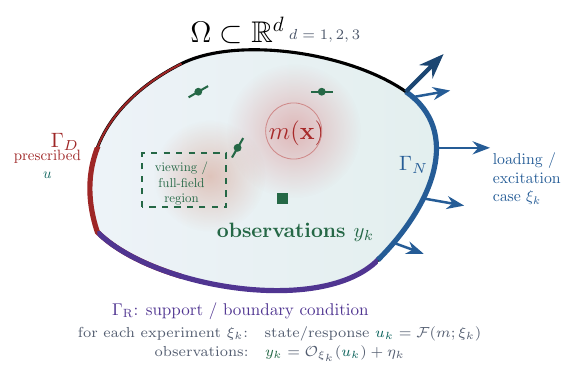}
\caption{General mechanics setting for the distributed-parameter inverse problem. 
An unknown spatial field \(m(\mathbf{x})\) enters a mechanics-governed forward model on a domain \(\Omega\). 
Loads, boundary conditions, excitations, sensors, and sampling choices define experiments \(\xi_k\), which generate state responses \(u_k\) and observations \(y_k\). 
The inverse problem uses these observations to infer the parameter field and to assess which spatial perturbation patterns are identifiable.}
\label{fig:general_setup}
\end{figure}

Relative to this literature, the contributions of the paper are fourfold. First, it
formulates a mechanics-oriented local information-operator framework for spatial
identifiability in distributed-parameter inverse problems, explicitly linking Fisher
information, Gauss--Newton curvature, sensitivity analysis, and parameter observability.
Second, it distinguishes pointwise visibility from operator-level spatial identifiability
by comparing diagonal information densities with metric- and prior-preconditioned
information modes. Third, it provides a unified treatment of heterogeneous observation
systems, including conditions for additive information, correlated errors, model
discrepancy, and nuisance-parameter confounding. Fourth, it illustrates the framework
through analytical and computational mechanics examples chosen to expose complementary
aspects of the same idea: closed-form beam kernels, static--dynamic data fusion, and
information-subspace reconstruction of a two-dimensional damage field.

\section{Bayesian distributed-parameter inverse problems}
\label{sec:bayesian}

\subsection{Forward model, prior, and posterior}

Let \(m\) denote an unknown spatial parameter field on a mechanical domain
\(\Omega \subset \mathbb{R}^d\), \(d=1,2,3\), or its finite-dimensional coefficient
representation after discretization. Depending on the application, \(m\) may represent
material, geometric, degradation, boundary/support, loading, environmental, or
model-discrepancy quantities. The notation \(m\) is used for both the field and its
coefficient vector when no ambiguity arises; metric-dependent statements are made
explicit below.

The parameter field enters a mechanics-governed forward problem through a state map
\[
u = \mathcal F(m),
\]
where \(u\) may denote displacement, strain, stress, modal response, frequency-response
data, or another state-dependent quantity. Observations are generated by an observation
operator \(\mathcal O_{\xi}\), indexed by an experiment design \(\xi\), giving the
parameter-to-observation map
\begin{equation}
H(m;\xi) = \mathcal O_{\xi}(\mathcal F(m)).
\label{eq:forward_map}
\end{equation}
The design variable \(\xi\) may include sensor positions, load locations, excitation
frequencies, sampling times, viewing regions, gauge orientations, or boundary-control
settings.

The measurement model is
\begin{equation}
y = H(m;\xi) + \eta,
\qquad
\eta \sim \mathcal N(0,R),
\label{eq:observation_model}
\end{equation}
where \(R\) is the observation-noise covariance. For a collection of experiments
\(\Xi=\{\xi_k\}_{k=1}^{K}\), the stacked observation map is
\begin{equation}
Y(m;\Xi)=
\begin{bmatrix}
H(m;\xi_1)\\
\vdots\\
H(m;\xi_K)
\end{bmatrix}.
\label{eq:stacked_map}
\end{equation}
In what follows, \(R\) denotes the covariance of the stacked data vector. It may be
block diagonal when experiments or sensor blocks are conditionally independent, or fully
populated when measurement errors are correlated across sensors, load cases, times, or
frequencies.

In the Bayesian setting, the unknown field is assigned a prior distribution
\(\pi_{\mathrm{pr}}(m)\), encoding mechanical or structural information available before
the data are collected~\cite{dashtiBayesianApproachInverse2017a}. For a Gaussian prior,
\begin{equation}
m \sim \mathcal N(m_{\mathrm{pr}},C_{\mathrm{pr}}),
\qquad
Q_{\mathrm{pr}} = C_{\mathrm{pr}}^{-1},
\label{eq:prior}
\end{equation}
the posterior density is
\begin{equation}
\pi(m\mid y)
\propto
\exp\!\left[-\frac{1}{2}\|y-Y(m;\Xi)\|_{R^{-1}}^2\right]\pi_{\mathrm{pr}}(m).
\label{eq:posterior}
\end{equation}
The negative log-posterior is
\begin{equation}
\mathcal J(m)
=
\Phi(m;y)
+
\frac{1}{2}\|m-m_{\mathrm{pr}}\|_{Q_{\mathrm{pr}}}^{2},
\qquad
\Phi(m;y)
=
\frac{1}{2}\|y-Y(m;\Xi)\|_{R^{-1}}^{2},
\label{eq:negative_log_posterior}
\end{equation}
where \(\Phi\) is the data-misfit functional. The maximum a posteriori estimate is \(m_{\mathrm{MAP}}=\operatorname*{arg\,min}_{m}\mathcal J(m)\).

\subsection{Local linearization and posterior geometry}

For distributed-parameter fields, full posterior exploration is often expensive because
each evaluation of \(Y(m;\Xi)\) requires a mechanics solve. A local alternative is to
study the posterior or likelihood geometry near a nominal parameter field \(m_0\), such
as a prior mean, a healthy baseline, a current design state, or a MAP estimate. The
linearized observation map is
\begin{equation}
Y(m_0+\delta m;\Xi)
\approx
Y(m_0;\Xi)+J(m_0;\Xi)\delta m,
\qquad
J(m_0;\Xi)=DY(m_0;\Xi).
\label{eq:linearization}
\end{equation}

Under this approximation and Gaussian noise, the likelihood is Gaussian in
\(\delta m\). With a Gaussian prior, the local posterior is Gaussian with precision
\begin{equation}
Q_{\mathrm{post}}(m_0)
\approx
Q_{\mathrm{pr}}+\mathcal I(m_0;\Xi),
\label{eq:local_post_precision}
\end{equation}
where \(\mathcal I(m_0;\Xi)\) is the data-induced precision correction defined below.

\Cref{fig:bayesian_geometry} illustrates this local update geometry in a
two-dimensional parameter subspace.
Thus the key question becomes not only what estimate minimizes \(\mathcal J\), but how
the data-induced precision changes the geometry of admissible parameter perturbations.

All statements in the following sections are local to \(m_0\), to the assumed covariance
model, and to the chosen experiment set \(\Xi\). This locality is essential: the framework
describes first-order distinguishability of distributed parameter perturbations, not
global identifiability over the full admissible parameter space.

\begin{figure}[t]
\centering
\includegraphics[width=\textwidth]{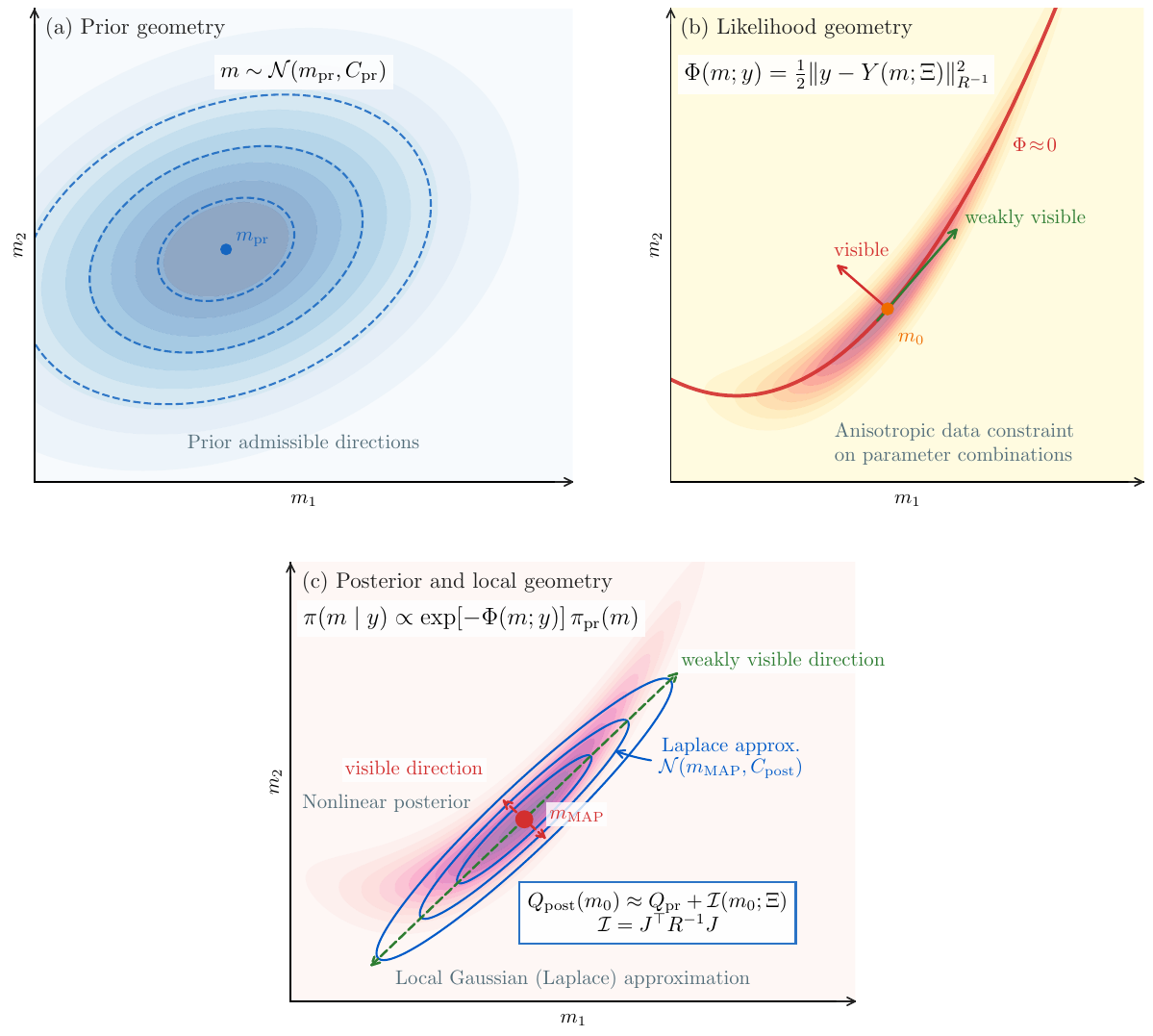}
\caption{Two-dimensional illustration of the Bayesian update geometry.
\textbf{(a)}~Prior distribution $\pi_{\mathrm{pr}}(m)=\mathcal N(m_{\mathrm{pr}},C_{\mathrm{pr}})$.
\textbf{(b)}~Likelihood $\pi(y\mid m)\propto\exp[-\Phi(m;y)]$ induced by a nonlinear forward map; weak and strong directions indicate data resolution anisotropy.
\textbf{(c)}~Posterior $\pi(m\mid y)$ and its local Laplace approximation (blue contours) around $m_{\mathrm{MAP}}$, with posterior precision given by \cref{eq:local_post_precision}.}

\label{fig:bayesian_geometry}

\end{figure}

\section{Local Bayesian information operator}
\label{sec:framework}

\subsection{Operator and information bilinear form}

\begin{definition}[Local information operator]
\label{def:local_information_operator}
Assume that the stacked observation covariance \(R\) is symmetric positive definite on
the measurement space. If singular covariance models, exact constraints, or projected
measurements are used, \(R^{-1}\) should be replaced by the corresponding generalized
inverse on the relevant measurement subspace. In a Hilbert-space formulation,
\(J(m_0;\Xi):\mathcal M\to\mathcal Y\) is the Fréchet derivative of the
parameter-to-observation map and the continuum operator is \(J^{\ast}R^{-1}J\), where
\(J^{\ast}\) denotes the adjoint with respect to the chosen parameter and observation
inner products. In the finite-dimensional notation used below, the
\emph{local information operator} at \(m_0\) for the experiment set \(\Xi\) is
\begin{equation}
\mathcal I(m_0;\Xi)
:=
J(m_0;\Xi)^{\!\top}R^{-1}J(m_0;\Xi).
\label{eq:local_information_operator}
\end{equation}
Here the transpose represents the corresponding discrete adjoint after the parameter
metric and observation covariance weighting have been specified.
\end{definition}

This definition is parameter-space based: observation sensitivities are pulled back
through \(J^{\!\top}R^{-1}\) to act on parameter-field perturbations. The operator
\(\mathcal I\) is the likelihood contribution to the local posterior precision. It is
positive semidefinite and defines an \emph{information-induced bilinear form} on parameter
perturbations,
\begin{equation}
\langle \delta m_1,\delta m_2\rangle_{\mathcal I}
=
\delta m_1^{\!\top}\mathcal I\,\delta m_2.
\label{eq:information_bilinear_form}
\end{equation}
The associated quadratic form
\(\delta m^{\!\top}\mathcal I\delta m=\|J\delta m\|_{R^{-1}}^2\)
measures the noise-weighted linearized output change induced by \(\delta m\), while the
bilinear form measures cross-information between perturbation directions.

\subsection{Fisher--Gauss--Newton--sensitivity-Gramian equivalence}
\label{sec:fisher_gn_gramian_equivalence}

Under the additive Gaussian observation model~\eqref{eq:observation_model}, with
parameter-independent covariance \(R\), the negative log-likelihood is
\begin{equation}
\Phi(m;y)
=
\frac{1}{2}\|y-Y(m;\Xi)\|_{R^{-1}}^{2}
+\mathrm{const}.
\label{eq:log_likelihood}
\end{equation}
The exact Hessian of this data-misfit functional at \(m_0\) is
\begin{equation}
\nabla^{2}\Phi(m_0)
=
J^{\!\top}R^{-1}J
+
\sum_{i=1}^{N_y}
\bigl[R^{-1}(Y(m_0;\Xi)-y)\bigr]_i\,\nabla^{2}Y_i(m_0;\Xi),
\label{eq:full_hessian}
\end{equation}
where \(J=DY(m_0;\Xi)\). The first term is the Gauss--Newton curvature; the second term
is residual dependent and contains the second derivatives of the parameter-to-observation
map. At small residuals this term is neglected in the Gauss--Newton approximation; in
expectation under the correctly specified local Gaussian model with data generated at
\(m_0\), it vanishes. The local curvature is then governed by the same operator defined in
\cref{eq:local_information_operator}.

This equivalence is classical in statistical parameter estimation, nonlinear
least-squares inverse problems, optimal experimental design, and observability or
identifiability analysis
\cite{jauffretObservabilityFisherInformation2007a,
powelEmpiricalObservabilityGramian2020,
boyaciogluDualityStochasticObservability2024,
ashyraliyevSystemsBiologyParameter2009,
qiOptimalExperimentalDesign2026}. It provides the operator
interpretation on which the subsequent spatial-identifiability diagnostics are built.

\begin{proposition}[Local Fisher--Gauss--Newton--sensitivity-Gramian equivalence]
\label{prop:operator_identity}
Assume the additive Gaussian observation model~\eqref{eq:observation_model}, with
covariance \(R\) symmetric positive definite and independent of \(m\), and let
\(J=DY(m_0;\Xi)\). For the locally linearized model, equivalently for the expected
curvature under the correctly specified local Gaussian model at \(m_0\), the local
information operator defined in \cref{eq:local_information_operator} has the following
equivalent interpretations:
\begin{enumerate}
\item it is the Fisher information matrix of the local Gaussian likelihood;
\item it is the Gauss--Newton Hessian of the data misfit, and the exact Hessian of the
locally linearized quadratic misfit;
\item it is a \emph{noise-weighted parameter-output sensitivity Gramian}, because for every
parameter perturbation \(\delta m\),
\begin{equation}
\|J\delta m\|_{R^{-1}}^{2}
=
\delta m^{\!\top}\mathcal I(m_0;\Xi)\,\delta m .
\label{eq:weighted_output_energy}
\end{equation}
\end{enumerate}
\end{proposition}

\begin{proof}
For the locally linearized observation map
\[
Y(m_0+\delta m;\Xi)
\approx
Y(m_0;\Xi)+J\delta m,
\]
the data misfit becomes
\[
\Phi_{\mathrm{lin}}(\delta m;y)
=
\frac{1}{2}
\|y-Y(m_0;\Xi)-J\delta m\|_{R^{-1}}^{2}
+\mathrm{const}.
\]
Differentiating twice with respect to \(\delta m\) gives
\[
\nabla_{\delta m}^{2}\Phi_{\mathrm{lin}}
=
J^{\!\top}R^{-1}J
=
\mathcal I(m_0;\Xi).
\]
Thus the operator is the Hessian of the locally linearized quadratic misfit and the
Gauss--Newton curvature of the nonlinear misfit.

For the Fisher-information interpretation, the score of the local Gaussian likelihood at
\(m_0\) is
\[
\mathfrak s(m_0;y)
=
J^{\!\top}R^{-1}\bigl(y-Y(m_0;\Xi)\bigr).
\]
Under the correctly specified local data-generating model at \(m_0\),
\[
\mathbb E\!\left[
\bigl(y-Y(m_0;\Xi)\bigr)
\bigl(y-Y(m_0;\Xi)\bigr)^{\!\top}
\right]
=
R .
\]
Hence
\[
\mathbb E\!\left[
\mathfrak s(m_0;y)\mathfrak s(m_0;y)^{\!\top}
\right]
=
J^{\!\top}R^{-1}RR^{-1}J
=
J^{\!\top}R^{-1}J
=
\mathcal I(m_0;\Xi).
\]
Finally, the linearized output perturbation associated with \(\delta m\) is
\(\delta y=J\delta m\). Therefore,
\[
\|\delta y\|_{R^{-1}}^{2}
=
(J\delta m)^{\!\top}R^{-1}(J\delta m)
=
\delta m^{\!\top}J^{\!\top}R^{-1}J\,\delta m
=
\delta m^{\!\top}\mathcal I(m_0;\Xi)\,\delta m .
\]
This proves the three interpretations.
\end{proof}

The sensitivity-Gramian interpretation is the mechanics-specific reading used below. A
perturbation \(\delta m\) is \emph{locally visible} when it produces a large noise-weighted
linearized output change \(\|J\delta m\|_{R^{-1}}\), \emph{weakly visible} when this quantity is
small, and \emph{locally invisible} when \(J\delta m=0\). After a parameter-space metric or
prior metric has been specified, the corresponding metric- or prior-preconditioned
eigenvectors define \emph{parameter-observability modes}: large eigenvalues indicate
strongly visible spatial patterns, small eigenvalues indicate weakly visible patterns,
and null modes identify locally unobservable perturbations under the selected experiment
set.

Thus \cref{prop:operator_identity} links the statistical, optimization, and mechanics
interpretations used in the remainder of the paper. This is the parameter-identification
sense of observability used throughout: a pullback information-metric analogy for
parameter perturbations, distinct from the classical control-theoretic
state-observability Gramian.

\subsection{Additivity, rank, and reparameterization}

If the stacked data consist of conditionally independent sensor, load-case, frequency, or
experiment blocks with Jacobians \(J_b\) and covariances \(R_b\), then
\begin{equation}
\mathcal I_{\mathrm{tot}}
=
\sum_b J_b^{\!\top}R_b^{-1}J_b.
\label{eq:additivity}
\end{equation}
Thus independent observation blocks contribute additively to the same parameter-space
information geometry. If the observation errors are correlated across blocks, the correct
expression is instead
\begin{equation}
\mathcal I_{\mathrm{tot}}
=
J_{\mathrm{tot}}^{\!\top}R_{\mathrm{joint}}^{-1}J_{\mathrm{tot}},
\label{eq:joint_information_operator}
\end{equation}
with the full joint covariance. Additivity is therefore a \emph{block-independence statement},
not a universal property of arbitrary data fusion; correlated-observation design requires
the covariance structure to be part of the information operator rather than a post-hoc
correction~\cite{attiaOptimalExperimentalDesign2022,qiOptimalExperimentalDesign2026}.

The operator is also parameterization dependent. Under a differentiable reparameterization
\(\tilde m=\mathcal R(m)\), with local Jacobian
\begin{equation}
T=\frac{\partial m}{\partial \tilde m}\bigg|_{\tilde m_0},
\label{eq:transform_matrix}
\end{equation}
the linearized map transforms as \(\tilde J=JT\), and hence
\begin{equation}
\tilde{\mathcal I}
=
T^{\!\top}\mathcal I T.
\label{eq:transformation_rule}
\end{equation}
Consequently, raw eigenvalues and eigenvectors are meaningful only relative to a stated
parameterization and metric.

Finally, for a discretized parameter field with \(n=\dim(m)\) and \(N_y\)
scalar observations, the discrete information matrix satisfies
\begin{equation}
\operatorname{rank}(\mathcal I)
\le
\min(N_y,n).
\label{eq:rank_bound}
\end{equation}

For continuous or distributed observations, the same bound applies after
discretization; in the operator setting, the effective rank is governed by the
smoothing properties of the forward map and the sensor resolution. Equality need not hold. Symmetries, smoothing by the forward operator, insufficient
excitation, poor sensor placement, and confounding with nuisance parameters can all reduce
the effective rank. Dense sensing alone does not remove physics-induced nullspaces.
The relevant object for distributed inverse problems is therefore the parameter-level
sensitivity Gramian \(\mathcal I\), not the classical state-observability Gramian.

\subsection{Metric-dependent information modes}

Because \(\mathcal I\) is positive semidefinite, its spectral structure provides a local
decomposition of parameter-field visibility. Large eigenvalues correspond to perturbation
directions that induce large noise-weighted linearized output changes; small eigenvalues
correspond to weakly visible directions; zero eigenvalues correspond to locally invisible
directions.

However, in finite-element and function-space settings, the raw Euclidean eigenproblem
\begin{equation}
\mathcal I\psi_i=\lambda_i\psi_i
\label{eq:raw_eigenproblem}
\end{equation}
is not intrinsically meaningful unless a parameter-space metric has been specified. If
\(m_h=\sum_j c_j\varphi_j\) is a finite-element representation and \(\mathcal I_h\) denotes the
discrete matrix representation of the information bilinear form, then an \(L^2\)-consistent
spectral problem is
\begin{equation}
\mathcal I_h c_i = \lambda_i M_h c_i,
\label{eq:mass_weighted_eigenproblem}
\end{equation}
where \(M_h\) is the mass matrix of the parameter basis. Other choices, such as energy
inner products or Sobolev metrics, lead to different but explicitly defined information
modes.

Thus raw modes should be regarded as \emph{coordinate diagnostics}. Metric-weighted modes are
appropriate when the goal is discretization-consistent spatial interpretation. In Bayesian
settings, the natural metric is often supplied by the prior precision, leading to the
prior-preconditioned formulation developed next.

The diagonal and the full operator provide complementary diagnostics (see
\cref{fig:diagonal_vs_full_operator}). For a discretized parameter field, the diagonal
entries of \(\mathcal I\) measure the visibility of coordinate perturbations in the
chosen basis. In a continuous representation, when the information operator admits a
kernel representation with a well-defined diagonal, \(\mathcal I(x,\bar x)\) defines a
two-point information object and its diagonal
\[
\mathcal I_{\mathrm{diag}}(x):=\mathcal I(x,x)
\]
gives a \emph{pointwise information density}. For ideal point observations or
distributional sensitivities, this density should be interpreted through the chosen
discretization, regularization, or finite sensor footprint. In this sense it is a
basis- or discretization-localized diagnostic, not an invariant continuum field
quantity. It is valuable for visualizing where localized perturbations are most visible. Spatial
identifiability, however, is governed by the full bilinear form
\[
\delta m^{\!\top}\mathcal I\delta m
\quad
\text{or}
\quad
\int_{\Omega}\int_{\Omega}
\delta m(x)\,\mathcal I(x,\bar x)\,\delta m(\bar x)\,dx\,d\bar x .
\]
The off-diagonal kernel encodes spatial coupling, cancellation, and
confounding between perturbations at different locations. Thus the diagonal
answers where localized visibility is high, while the full operator answers
which spatial patterns are locally distinguishable.

\begin{figure}[t]
\centering
\includegraphics[width=\textwidth]{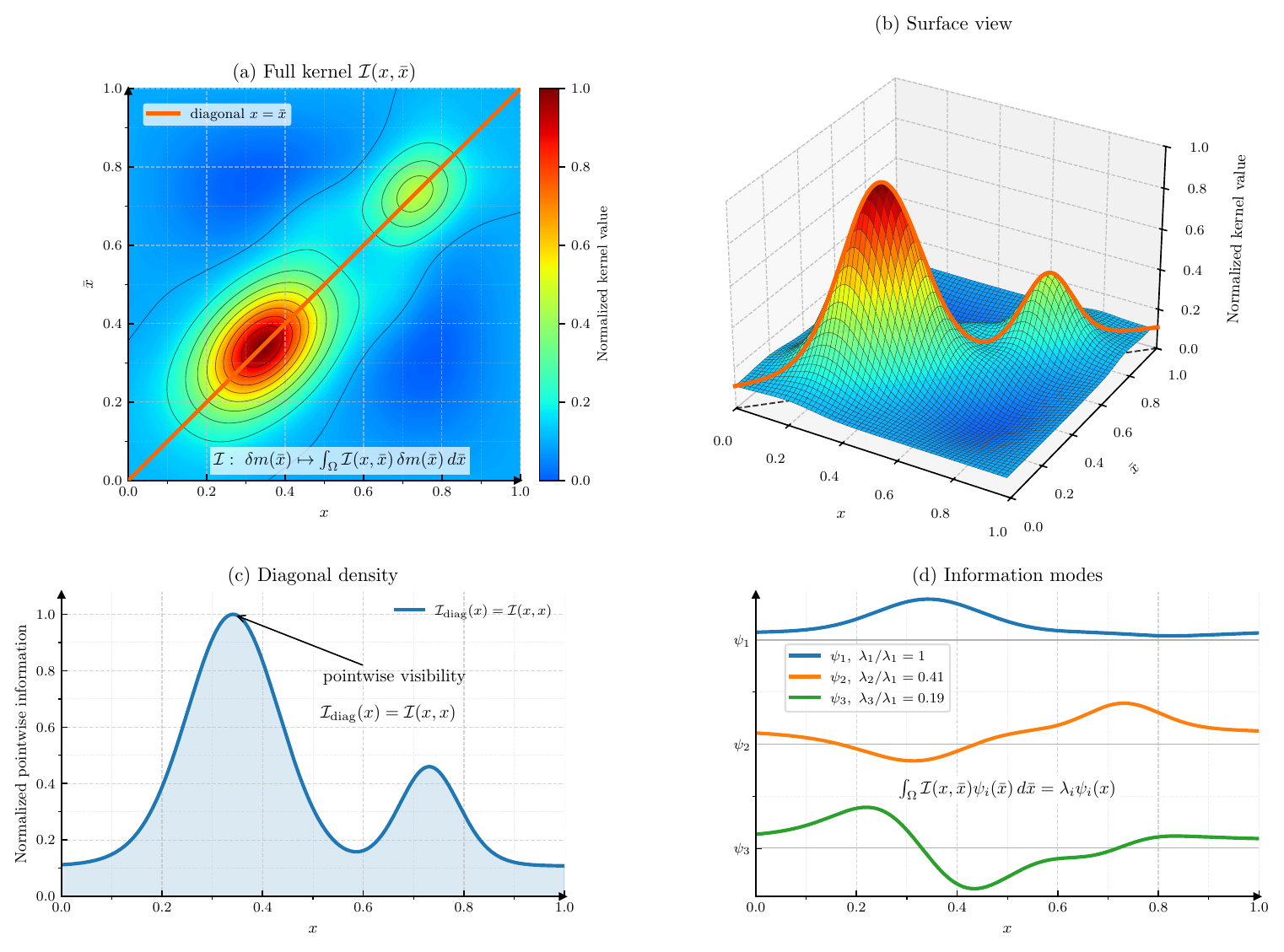}
\caption{Conceptual illustration of pointwise visibility and operator-level spatial identifiability.
(a) Full information kernel \(\mathcal I(x,\bar x)\), with the diagonal \(x=\bar x\) marked.
(b) Surface view of the same kernel, emphasizing that information is a two-point spatial object.
(c) Diagonal information density \(\mathcal I(x,x)\), which visualizes coordinate-wise local visibility.
(d) Leading information modes \(\psi_i(x)\), obtained from the full operator eigenproblem. The diagonal is useful for visualization, but spatial identifiability is governed by the full kernel and its metric- or prior-preconditioned eigenmodes. The modes in panel~(d) are vertically shifted for visualization.}
\label{fig:diagonal_vs_full_operator}
\end{figure}

\section{Prior-preconditioned identifiable subspaces}
\label{sec:prior_preconditioned}

\subsection{Posterior precision and prior-relative modes}

Under the local Gaussian approximation, the posterior precision is given by
\cref{eq:local_post_precision}. This additive prior-plus-likelihood
precision structure is exact for linear Gaussian inverse problems and is standard
in local Gaussian or Laplace approximations of nonlinear Bayesian inverse problems
\cite{stuartInverseProblemsBayesian2010,dashtiBayesianApproachInverse2017a,
bui-thanhComputationalFrameworkInfiniteDimensional2013,
petraComputationalFrameworkInfiniteDimensional2014}.

The prior contributes \(Q_{\mathrm{pr}}\), while the likelihood contributes
\(\mathcal I\). Thus \(\mathcal I\) should not only be interpreted as an absolute
measure of data sensitivity, but also relative to the prior precision that already
regularizes or constrains the parameter field.

For large mechanics models, prior-preconditioned diagnostics require scalable actions of
the prior or whitening operator in addition to \(J\) and \(J^{\!\top}\). When the prior is
defined through a differential operator, these actions may involve additional linear
solves; this computational requirement is included in the matrix-free viewpoint of
\cref{sec:matrix_free}.

To compare likelihood information with prior information, we use the standard
prior-preconditioned Hessian construction underlying likelihood-informed subspaces \cite{flathFastAlgorithmsBayesian2011,
bui-thanhComputationalFrameworkInfiniteDimensional2013,
cuiLikelihoodinformedDimensionReduction2014,
spantiniOptimalLowrankApproximations2015}.
The prior-relative modes are obtained from the generalized eigenproblem
\begin{equation}
\mathcal I\phi_i
=
\lambda_i^{\mathrm{pr}} Q_{\mathrm{pr}}\phi_i .
\label{eq:generalized_eigenproblem}
\end{equation}
In the terminology of the present paper, these modes are interpreted as
\emph{prior-relative parameter-observability modes}: they identify spatial
parameter-field perturbations for which likelihood precision is large relative to prior
precision.

Equivalently, when \(C_{\mathrm{pr}}\) is symmetric positive definite and using a
symmetric square root \(C_{\mathrm{pr}}^{1/2}\), the same generalized eigenvalues are
obtained from the ordinary eigenproblem for the prior-preconditioned information
operator, or prior-preconditioned Gauss--Newton data-misfit Hessian
\cite{flathFastAlgorithmsBayesian2011,
bui-thanhComputationalFrameworkInfiniteDimensional2013,
spantiniOptimalLowrankApproximations2015},

\begin{equation}
C_{\mathrm{pr}}^{1/2}\mathcal I C_{\mathrm{pr}}^{1/2} w_i
=
\lambda_i^{\mathrm{pr}} w_i,
\qquad
\phi_i=C_{\mathrm{pr}}^{1/2}w_i .
\label{eq:prior_preconditioned_eigenproblem}
\end{equation}
Here \(w_i\) denotes a mode in prior-whitened coordinates, while \(\phi_i\) denotes the
corresponding perturbation mode in the original parameter coordinates. We use
\(\phi_i\) for these prior-relative modes to distinguish them from the raw or
metric-weighted information modes \(\psi_i\) introduced above.

Here the eigenpairs are used primarily as mechanics-oriented diagnostics rather than as a
reduced basis for posterior sampling. Each mode \(\phi_i\) represents a prior-plausible
parameter-field pattern, such as stiffness, compliance, or damage, whose visibility
relative to the prior is quantified by \(\lambda_i^{\mathrm{pr}}\). The formulation also
reduces sensitivity to mesh-dependent parameter scalings and connects the diagnostics to
optimal low-rank posterior approximation and likelihood-informed-subspace theory
\cite{cuiLikelihoodinformedDimensionReduction2014,
spantiniOptimalLowrankApproximations2015,
cuiPriorNormalizationCertified2022a}.

The eigenvalues \(\lambda_i^{\mathrm{pr}}\) have the Rayleigh-quotient interpretation
\begin{equation}
\lambda_i^{\mathrm{pr}}
=
\frac{\phi_i^{\!\top}\mathcal I\phi_i}
     {\phi_i^{\!\top}Q_{\mathrm{pr}}\phi_i},
\label{eq:prior_relative_rayleigh}
\end{equation}
for nonzero \(\phi_i\). They therefore measure the ratio between likelihood
information and prior information in each mode. For the linearized Gaussian problem, if
\(c_i\) denotes the modal coefficient associated with a \(Q_{\mathrm{pr}}\)-orthonormal
generalized eigenmode \(\phi_i\), its posterior variance is reduced relative to its
prior variance by
\begin{equation}
\frac{\operatorname{Var}_{\mathrm{post}}(c_i)}
     {\operatorname{Var}_{\mathrm{pr}}(c_i)}
=
\frac{1}{1+\lambda_i^{\mathrm{pr}}}.
\label{eq:posterior_variance_reduction}
\end{equation}
Thus \(\lambda_i^{\mathrm{pr}}\gg 1\) indicates a data-dominated direction, whereas
\(\lambda_i^{\mathrm{pr}}\ll 1\) indicates a prior-dominated direction.

\subsection{Likelihood-informed subspaces}

Following the likelihood-informed-subspace (LIS) viewpoint
\cite{cuiLikelihoodinformedDimensionReduction2014,
spantiniOptimalLowrankApproximations2015,
cuiPriorNormalizationCertified2022a}, the dominant prior-relative modes define
\begin{equation}
\mathcal V_{\mathrm{LIS}}
=
\operatorname{span}\{\phi_i:\lambda_i^{\mathrm{pr}}>\tau\},
\label{eq:lis_definition}
\end{equation}
where \(\tau\) is a problem-dependent threshold. This subspace contains the
prior-plausible directions in which the likelihood produces substantial posterior
contraction. Its dimension is often much smaller than the discretized parameter
dimension, even when the parameter field itself is high dimensional.

Here \(\mathcal V_{\mathrm{LIS}}\) is a local Laplace-type informed subspace defined by
the single information operator at \(m_0\). In nonlinear Bayesian inverse problems, LIS
constructions may instead use posterior-averaged or certified likelihood curvature over
the posterior support. The present subspace therefore serves as a local mechanics
diagnostic that complements those broader nonlinear LIS definitions.

The complement of \(\mathcal V_{\mathrm{LIS}}\) is not necessarily unimportant. It
contains directions that remain controlled primarily by the prior under the given
observation programme. In mechanics applications, these directions may correspond to
poorly excited regions, damage patterns hidden from the current sensor layout, spatial
oscillations smoothed by the forward operator, or parameter combinations confounded with
nuisance quantities. In the context of sensor or experiment design, these weak directions
are natural candidates for additional observations.

The threshold \(\tau\) should be chosen according to the intended diagnostic purpose. By
\cref{eq:posterior_variance_reduction}, \(\tau=1\) selects directions whose posterior
variance is reduced by at least one half relative to the prior, while \(\tau=3\) selects
directions reduced to at most one quarter. For exploratory diagnostics, a more inclusive
threshold such as \(\tau=10^{-2}\) can be useful to retain weakly informed but
non-negligible directions. In practice, the eigenvalue decay and the presence of a
spectral elbow should also guide the retained subspace dimension.

\subsection{Prior-whitened sensitivity-Gramian interpretation}

The prior-preconditioned operator also has a direct sensitivity-Gramian interpretation.
For a Gaussian prior, standardized prior perturbations are mapped to physical
parameter-field perturbations by
\begin{equation}
\delta m=C_{\mathrm{pr}}^{1/2}\delta\nu .
\label{eq:prior_whitened_parameter}
\end{equation}
Equivalently, a prior draw can be written as
\(m=m_{\mathrm{pr}}+C_{\mathrm{pr}}^{1/2}\nu\) with
\(\nu\sim\mathcal N(0,I)\). Thus \(C_{\mathrm{pr}}^{1/2}\) maps standardized prior
perturbations into prior-plausible perturbations of the physical parameter field.

A perturbation \(\delta \nu\) in prior-whitened coordinates induces
\[
\delta y=JC_{\mathrm{pr}}^{1/2}\delta \nu .
\]
The corresponding noise-weighted output energy is
\begin{equation}
\|\delta y\|_{R^{-1}}^{2}
=
\delta \nu^{\!\top}
\left[
C_{\mathrm{pr}}^{1/2}J^{\!\top}R^{-1}JC_{\mathrm{pr}}^{1/2}
\right]
\delta \nu .
\label{eq:prior_whitened_gramian}
\end{equation}
Therefore,
\begin{equation}
C_{\mathrm{pr}}^{1/2}\mathcal I C_{\mathrm{pr}}^{1/2}
=
C_{\mathrm{pr}}^{1/2}J^{\!\top}R^{-1}JC_{\mathrm{pr}}^{1/2}
\label{eq:prior_whitened_information_operator}
\end{equation}
is the noise-weighted parameter-output sensitivity Gramian in prior-whitened coordinates,
used here as an observability-type diagnostic for prior-plausible parameter-field
perturbations.

The raw information operator \(\mathcal I\) measures the visibility of arbitrary
coordinate perturbations in the chosen parameter representation. The prior-whitened
sensitivity Gramian measures the visibility of perturbations expressed in units of
prior uncertainty. Consequently, its dominant eigenmodes identify the mechanical
field patterns that are both prior-plausible and strongly distinguished by the
observation system.

\section{Heterogeneous observation operators}
\label{sec:observations}

\subsection{Chain-rule factorization}

To treat heterogeneous observations within one algebra, write the observation process in
the abstract form
\begin{equation}
y
=
\mathcal O_{\xi}
\bigl(u,\nabla u,\nabla^2u,\ldots;\mathbf x,t\bigr),
\qquad
u=\mathcal F(m).
\label{eq:general_observation_operator}
\end{equation}
This covers pointwise displacements, rotations, and strains; distributed strain and
full-field measurements; modal quantities; influence-line samples; time histories; and
frequency-response data. The observation technology changes the functional
\(\mathcal O_{\xi}\), the dimension of \(y\), and the covariance \(R\), but not the local
information template. This common template is useful for multi-type structural sensing and
multiresponse model updating, where different observation classes may carry complementary
information about the same parameter
field~\cite{bertolaOptimalMultiTypeSensor2017,bellMultiresponseParameterEstimation2007,
kimSequentialFrameworkImproving2019}.

The derivative of \(H=\mathcal O_{\xi}\circ\mathcal F\) factors as
\begin{equation}
J = S_{\xi}F_m,
\qquad
F_m=D\mathcal F(m_0),
\qquad
S_{\xi}=D\mathcal O_{\xi}(u_0),
\label{eq:chain_rule_factorization}
\end{equation}
where \(u_0=\mathcal F(m_0)\). In PDE settings, \(F_m\) and \(S_{\xi}\) should be
interpreted as Fréchet derivatives or their discrete matrix representations. Substituting
\cref{eq:chain_rule_factorization} into \cref{eq:local_information_operator}, the
information operator becomes
\begin{equation}
\mathcal I
=
F_m^{\!\top}S_{\xi}^{\!\top}R^{-1}S_{\xi}F_m.
\label{eq:factorized_information}
\end{equation}

This factorization separates three mechanisms:
\begin{enumerate}
\item \(F_m\): how parameter perturbations propagate through the mechanics;
\item \(S_{\xi}\): how the observation system samples the mechanical state;
\item \(R^{-1}\): how the statistical model weights parameter-induced output changes.
\end{enumerate}
This separation is important for interpreting \emph{observation complementarity}: two
observation systems may be different because they sample or process different state
components through \(S_{\xi}\), or because the mechanics map \(F_m\) makes them sensitive
to different parameter-field patterns.

\subsection{Observation classes}

\Cref{fig:observation_classes_conceptual} illustrates representative observation
mechanisms, and \cref{tab:observation_taxonomy} groups them by local Jacobian structure
rather than sensor technology. The classification is illustrative, not exhaustive or
mutually exclusive: the same physical sensor may enter different classes depending on
whether its output is used as a point value, a time history, a frequency-domain feature, a
modal quantity, or an averaged or derived measurement. The purpose is to show how different
observation blocks enter the same information-operator construction while inducing
different geometry in parameter space.

\begin{figure}[!htbp]
\centering
\includegraphics[width=\textwidth]{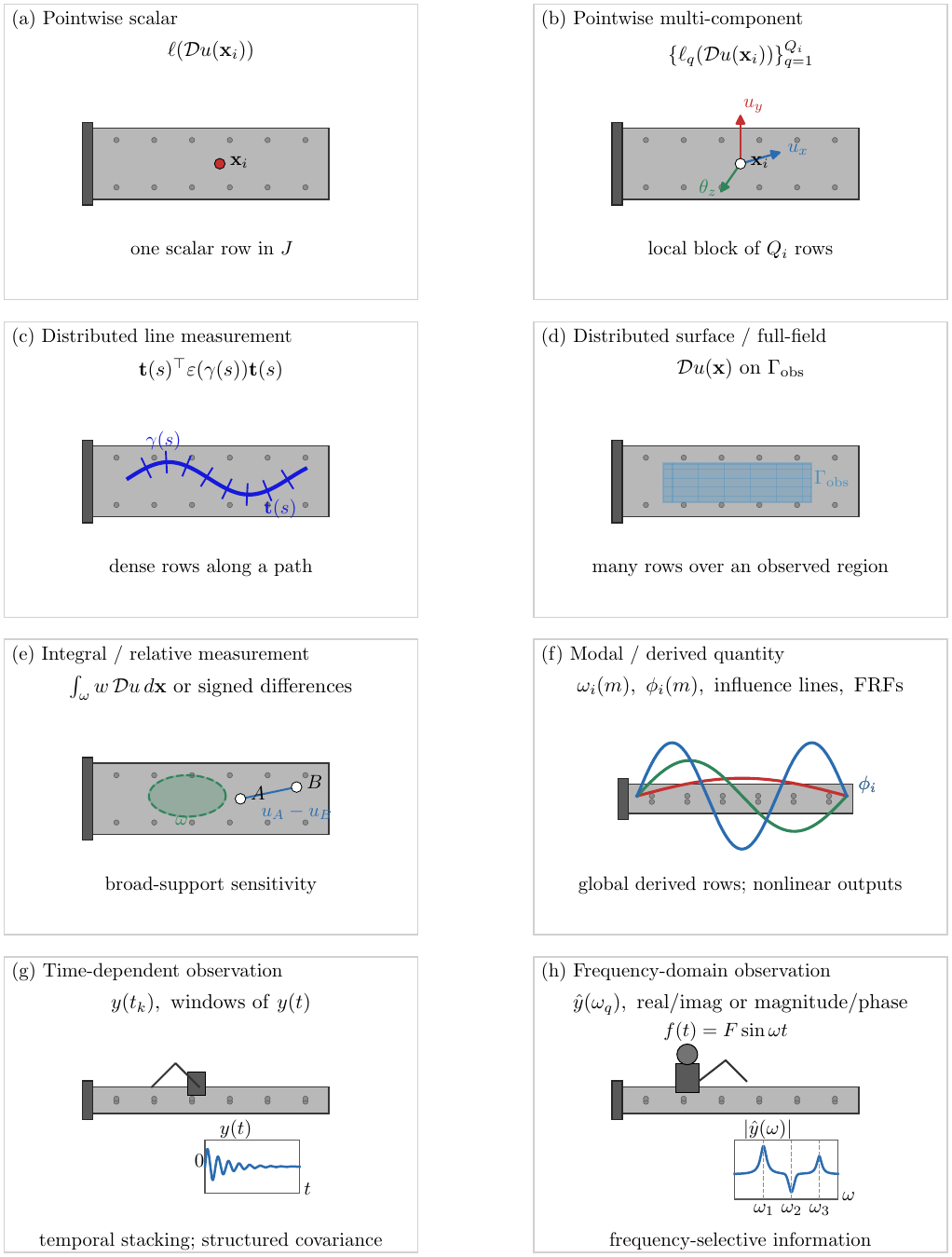}
\caption{Conceptual examples of observation classes in the local information-operator framework. 
Each panel illustrates a different way in which the mechanical state \(u\), or a derived quantity of \(u\), can be sampled to form rows of the local Jacobian \(J\). 
Pointwise measurements contribute localized rows, distributed line and surface measurements provide spatially extended rows, integral and relative measurements produce broad-support sensitivities, and modal, influence-line, time-domain, and frequency-domain observations contribute derived or stacked rows. 
All observation classes enter the same local information operator, but their Jacobian rows induce different geometry in the informed parameter subspace.}
\label{fig:observation_classes_conceptual}
\end{figure}

\begin{table}[!htbp]
\footnotesize
\centering
\caption{Representative observation classes grouped by local Jacobian structure.}
\label{tab:observation_taxonomy}

\newcommand{\tabragged}{\raggedright\arraybackslash}
\begin{tabularx}{\textwidth}{
	>{\tabragged}p{0.13\textwidth}
	>{\tabragged}p{0.24\textwidth}
	>{\tabragged}p{0.22\textwidth}
>{\tabragged}X}
\toprule
Observation class & Example sensors or data & Local observation form & Effect on \(\mathcal I\) \\
\midrule

Pointwise scalar &
LVDT, laser displacement sensor, single strain gauge, tiltmeter, accelerometer channel &
\(\ell(\mathcal D u(\mathbf x_i))\) &
One scalar; rank-one contribution under independent scalar noise; localized and direction-selective.\\

\addlinespace[3pt]
\midrule

Pointwise multi-component &
Triaxial accelerometer, rosette strain gauge, multi-axis displacement sensor, local displacement--rotation measurement &
\(\{\ell_q(\mathcal D u(\mathbf x_i))\}_{q=1}^{Q_i}\) &
\(Q_i\) rows; local block contribution resolving multiple state components at one location. \\

\addlinespace[3pt]
\midrule

Distributed line measurement &
Distributed fiber-optic strain sensing, long-gauge strain sensing, scanned line displacement/strain measurement &
\(\mathbf t(s)^{\!\top}\varepsilon(\gamma(s))\mathbf t(s)\) &
After spatial sampling, \(N_\gamma\) rows; extended coverage along a path, directionally constrained by gauge orientation. \\

\addlinespace[3pt]
\midrule

Distributed surface or full-field data &
Digital image correlation, laser scanning, photogrammetry, full-field displacement or strain mapping&
\(\mathcal D u(\mathbf x)\) on \(\Gamma_{\mathrm{obs}}\) &
After spatial sampling, \(qN_{\mathrm{obs}}\) rows; potentially high-rank observation, still limited by mechanics and noise. \\

\addlinespace[3pt]
\midrule

Integral or relative measurement &
Load-cell resultant, average strain gauge, relative displacement transducer, crack-opening gauge, bridge weigh-in-motion response &
\(\int_{\Omega_{\mathrm{obs}}} w\,\mathcal D u\,d\mathbf x\), or signed differences &
Broad-support sensitivity patterns; useful for global or averaged features. \\

\addlinespace[3pt]
\midrule

Modal or derived quantity &
Modal frequencies, mode shapes, modal curvatures, influence lines, identified flexibility or stiffness features &
\(\omega_i(m)\), \(\phi_i(m)\), modal curvatures, influence-line ordinates &
Stacked rows; global sensitivities; nonlinear output quantity requiring local linearization. \\

\addlinespace[3pt]
\midrule

Time-dependent observation &
Accelerometer time histories, strain time histories, displacement records, transient impact or moving-load responses &
\(y(t_k)\), windows of \(y(t)\), transient response &
\(N_t\) rows; temporal stacking; covariance may be strongly structured. \\

\addlinespace[3pt]
\midrule

Frequency-domain observation &
Frequency-response functions, transfer functions, spectral amplitudes, operational modal spectra &
\(\hat y(\omega_q)\), real--imaginary or magnitude--phase data &
\(N_\omega\) rows; frequency-selective information; modal and antiresonance structure affects visibility. \\

\bottomrule
\end{tabularx}
\end{table}

One scalar observation contributes one row to \(J\) and therefore one rank-one information
term. Distributed and full-field observations contribute many rows, but they do not
automatically eliminate nullspaces created by the mechanics, the parameterization, or the
noise model. The observation class changes not only the amount of information but also the
geometry of the informed parameter subspace.

\subsection{Complementarity, redundancy, and weak-direction gain}

For an observation block \(b\), define
\begin{equation}
\mathcal I_b = J_b^{\!\top}R_b^{-1}J_b.
\label{eq:sensor_information}
\end{equation}
Its locally informed subspace is
\begin{equation}
\mathcal V_b
:=
\operatorname{range}(\mathcal I_b)
=
\operatorname{range}(J_b^{\!\top}).
\label{eq:informed_subspace}
\end{equation}
For conditionally independent blocks,
\begin{equation}
\mathcal V_{\mathrm{tot}}
=
\operatorname{range}(\mathcal I_{\mathrm{tot}})
=
\sum_b \mathcal V_b.
\label{eq:total_subspace}
\end{equation}
Adding an observation block can enlarge the informed subspace or strengthen directions
already present, but it cannot remove information in the positive-semidefinite ordering.

This motivates three subspace-level diagnostics:
\begin{itemize}[label=--]
\item \emph{Complementarity}: a candidate observation is complementary when its
information contribution has substantial projection onto directions weakly informed by the
existing observation set.
\item \emph{Redundancy}: a candidate observation is redundant when its dominant
information lies mostly in directions already strongly informed.
\item \emph{Weak-direction gain}: the value of a new observation is measured by the
information it adds in currently weak directions, not only by its total trace or
determinant contribution.
\end{itemize}

The idea is close in spirit to goal-oriented and Bayesian experimental design, where the
value of an experiment is judged through its effect on a target uncertainty or expected
utility~\cite{attiaGoalorientedOptimalDesign2018,huanOptimalExperimentalDesign2024}.
Here, however, the target is a diagnostic subspace of the existing local information
operator: the question is whether a candidate observation improves currently weak
mechanics-field patterns.

Let \(\Pi_{\mathrm{weak}}\) denote the projector onto a chosen weak subspace of the current
information operator, defined with respect to the relevant parameter-space metric or prior.
For example, \(\Pi_{\mathrm{weak}}\) may be constructed from the eigenvectors associated
with small nonzero eigenvalues of the metric-weighted or prior-preconditioned operator.
In a prior-preconditioned setting, this may be implemented by selecting modes with
\(\lambda_i^{\mathrm{pr}}<\tau_{\mathrm{weak}}\), or by taking the complement of the retained
likelihood-informed subspace, with the projector formed in the corresponding
prior-whitened or \(Q_{\mathrm{pr}}\)-orthonormal coordinates.
The following expressions should be read in those coordinates, or with the corresponding
metric-adjoint projector in the original parameter coordinates.
The weak-direction gain operator of a candidate block \(b\) relative to an existing
observation set \(\mathcal S\) is
\begin{equation}
G_{\mathrm{weak}}(b\mid\mathcal S)
=
\Pi_{\mathrm{weak}}\mathcal I_b\Pi_{\mathrm{weak}}.
\label{eq:weak_direction_gain}
\end{equation}
A scalar summary is
\begin{equation}
g_{\mathrm{weak}}(b\mid\mathcal S)
=
\operatorname{tr}\!\left(G_{\mathrm{weak}}(b\mid\mathcal S)\right).
\label{eq:weak_direction_scalar_gain}
\end{equation}
This trace is understood in the same metric or prior-whitened coordinates used to define
\(\Pi_{\mathrm{weak}}\); otherwise its numerical value is coordinate dependent.

In Bayesian design settings, the same construction can be applied to the blockwise
analogue of the prior-preconditioned operator in
\cref{eq:prior_preconditioned_eigenproblem}, namely
\(C_{\mathrm{pr}}^{1/2}\mathcal I_b C_{\mathrm{pr}}^{1/2}\), so that weak-direction gain
is measured relative to prior-plausible perturbations. These diagnostics do not replace
classical A-, D-, or E-optimal design criteria; they expose whether a candidate
observation fills current information gaps or mainly reinforces already visible
directions.

If \(\mathcal I\) is positive definite, it defines the information ellipsoid
\begin{equation}
\mathcal E(\mathcal I)
=
\{\delta m:\delta m^{\!\top}\mathcal I\delta m\le 1\}.
\label{eq:information_ellipsoid}
\end{equation}
If \(\mathcal I\) is singular, the same set is an unbounded ellipsoidal cylinder whose
unbounded directions coincide with the local nullspace. A complementary observation
contracts the cylinder in directions where it is long; a redundant observation mainly
contracts directions already short. This ellipsoidal geometry is the geometric basis of
the proposed complementarity diagnostics.

\section{Discrepancy and nuisance parameters}
\label{sec:discrepancy}

\subsection{Covariance inflation}

Real mechanics models may differ from the physical system because of omitted physics,
uncertain boundary conditions, environmental variation, discretization error, load
uncertainty, or modelling assumptions. Following standard Bayesian model-discrepancy
treatments~\cite{kennedyBayesianCalibrationComputer2001,brynjarsdottirLearningPhysicalParameters2014},
a simple Gaussian discrepancy model can be written as

\begin{equation}
y=H(m;\xi)+\delta(\xi)+\eta,
\qquad
\eta\sim\mathcal N(0,R),
\qquad
\delta\sim\mathcal N(0,C_{\delta}),
\label{eq:discrepancy_model}
\end{equation}
with \(\delta\) independent of \(\eta\). If the discrepancy is treated as zero-mean
Gaussian, the effective covariance is
\begin{equation}
R_{\mathrm{eff}}=R+C_{\delta},
\label{eq:effective_covariance}
\end{equation}

and the information operator becomes
\begin{equation}
\mathcal I_{\mathrm{eff}}
=
J^{\!\top}R_{\mathrm{eff}}^{-1}J.
\label{eq:effective_information}
\end{equation}
\emph{Covariance inflation} reduces information in observation-space directions where model
discrepancy is large, accounting for additional uncertainty without eliminating systematic
bias or confounding.

This interpretation assumes a zero-mean discrepancy independent of the observation noise.
If model error is systematic or strongly correlated with the physical parameter field, for
example through unmodelled boundary friction, temperature effects, or support
nonlinearities, covariance inflation may only broaden the posterior. Such cases require
explicit discrepancy or nuisance-parameter modelling. The discrepancy representation is
itself consequential: overly flexible or poorly constrained discrepancy priors can
introduce additional confounding between physical parameters and model-error
terms~\cite{brynjarsdottirLearningPhysicalParameters2014,lingSelectionModelDiscrepancy2014a}.

\subsection{Schur-complement information loss}

A second mechanism is explicit \emph{nuisance-parameter coupling}. Let \(m\) denote the
parameter field of interest and \(n\) collect nuisance quantities such as load amplitudes,
support stiffnesses, temperature effects, mass-density uncertainty, or sensor calibration
parameters. Depending on the engineering question, such quantities may instead be inferred
directly. With local Jacobians \(J_m\) and \(J_n\), the joint information matrix is
\begin{equation}
\mathcal I_{\mathrm{joint}}
=
\begin{bmatrix}
J_m^{\!\top}R^{-1}J_m & J_m^{\!\top}R^{-1}J_n\\
J_n^{\!\top}R^{-1}J_m & J_n^{\!\top}R^{-1}J_n
\end{bmatrix}
=
\begin{bmatrix}
\mathcal I_{mm} & \mathcal I_{mn}\\
\mathcal I_{nm} & \mathcal I_{nn}
\end{bmatrix}.
\label{eq:joint_information}
\end{equation}

When the nuisance variables are unknown and are locally eliminated or marginalized in the
Gaussian approximation, the effective information for \(m\) is the \emph{Schur complement}
\begin{equation}
\mathcal I_{m\mid n}
=
\mathcal I_{mm}
-
\mathcal I_{mn}\mathcal I_{nn}^{\dagger}\mathcal I_{nm}.
\label{eq:schur_complement}
\end{equation}
Here \(\mathcal I_{nn}^{\dagger}\) denotes a generalized inverse when the nuisance block
is singular. In that case, \cref{eq:schur_complement} is a generalized Schur complement:
it should be used when the relevant range condition holds, or after restricting \(n\) to
its locally identifiable subspace. Equivalently, it represents the information remaining
after projecting out locally identifiable nuisance-sensitive directions.
With Gaussian priors on \(m\) and \(n\), the same operation should be applied to the
joint posterior precision rather than to the likelihood information alone.

When the Schur complement is well defined,
\begin{equation}
\mathcal I_{m\mid n}\preceq \mathcal I_{mm},
\label{eq:information_loss_order}
\end{equation}
so nuisance coupling cannot increase the information available for \(m\). Unlike simple
covariance inflation, however, nuisance elimination can rotate the identifiable subspace:
directions of \(m\) strongly aligned with nuisance sensitivities lose information
selectively. This provides a local diagnostic for confounding between damage, loading,
boundary conditions, environmental effects, and material parameters.

For example, a localized stiffness-loss field may be partially confounded with uncertain
load amplitude, support flexibility, or temperature-induced global softening if these
nuisance quantities produce observation-space sensitivities aligned with the damage
sensitivities. In such cases, the off-diagonal block \(\mathcal I_{mn}\) removes from the
damage field precisely those components that cannot be distinguished from the nuisance
effects by the chosen observations.

\section{Matrix-free operator evaluation}
\label{sec:matrix_free}

\subsection{Operator actions}
For large finite-element and PDE-constrained mechanics models, the local
information operator need not be assembled explicitly. The required
computational primitive is the \emph{matrix-free action}
\begin{equation}
v
\mapsto
\mathcal I v
=
J^{\!\top}\!\left(R^{-1}(Jv)\right).
\label{eq:matrix_free_action}
\end{equation}

Here \(v\) is a parameter-field perturbation, \(Jv\) is the corresponding linearized
observation perturbation, and \(J^{\!\top}\) maps a weighted observation-space quantity
back to parameter space. The action can be evaluated whenever tangent and adjoint actions
of the parameter-to-observation map are available. In PDE and finite-element mechanics,
\(Jv\) is obtained from a tangent solve, \(R^{-1}(Jv)\) is a covariance-weighting
operation, and \(J^{\!\top}w\) is obtained from an adjoint solve. In differentiable
computational graphs, the same action is computed by Jacobian-vector and vector-Jacobian
products.

The matrix-free action itself is standard in PDE-constrained optimization,
Bayesian inverse problems, and optimal experimental design
\cite{flathFastAlgorithmsBayesian2011,alexanderianFastScalableMethod2016,
antilFisherInformationBasedSensorPlacement2026}. Here it provides scalable access to
information modes, prior-relative modes, weak-direction gain, trace estimates, and
diagonal information-density estimates. The sensitivity matrix need not be stored; memory
requirements are governed by state, adjoint, and covariance-solve operations rather than
by the product of the number of parameter degrees of freedom and the number of
observations.

\subsection{Eigenmode extraction}

Once the action~\eqref{eq:matrix_free_action} is available, Krylov, Lanczos,
randomized singular value decomposition (SVD), or randomized eigensolvers can be used to
extract dominant information modes without forming the full matrix. For sensor or
experiment design, the same actions compare candidate observation blocks through their
effect on leading modes, weak modes, trace estimates, log-determinant approximations, or
posterior variance reduction.

The same matrix-free machinery can also be used for weak-direction gain. Given a basis
\(W\) for a weak subspace, the projected candidate contribution is
\begin{equation}
W^{\!\top}\mathcal I_b W,
\label{eq:weak_projected_matrix}
\end{equation}
which can be evaluated using only actions of the candidate-block operator \(\mathcal I_b\)
on the weak basis vectors, avoiding full assembly for every candidate sensor or
experiment.

When only diagonal information densities are required, they can be obtained by assembling
local sensitivity contributions where feasible, or estimated matrix-free using probing or
randomized diagonal estimators applied to \(\mathcal I\). They remain
coordinate-dependent diagnostics; coupling and cancellation require the operator or
spectral analyses above.

\subsection{Prior-preconditioned actions}

For prior-preconditioned modes, the required action is the matrix-free form of the
operator in
\cref{eq:prior_preconditioned_eigenproblem,eq:prior_whitened_information_operator}:
\begin{equation}
v
\mapsto
C_{\mathrm{pr}}^{1/2}J^{\!\top}R^{-1}JC_{\mathrm{pr}}^{1/2}v.
\label{eq:prior_preconditioned_action}
\end{equation}
This requires, in addition to tangent and adjoint actions, the ability to apply
\(C_{\mathrm{pr}}^{1/2}\) or an equivalent prior-whitening transformation. When a prior
precision \(Q_{\mathrm{pr}}\) is more accessible than \(C_{\mathrm{pr}}^{1/2}\), the
generalized eigenproblem~\eqref{eq:generalized_eigenproblem} can be solved directly using
operators for \(\mathcal I\) and \(Q_{\mathrm{pr}}\).

Reduced-order models, surrogate models, or automatic-differentiation implementations may
reduce the cost of repeated tangent and adjoint actions. These approximations change how
the action is evaluated, not the diagnostic object itself: the approximated quantity
remains the local information operator and its metric- or prior-preconditioned variants.

In very large three-dimensional models, overall scalability also depends on
the mechanics solver, adjoint implementation, covariance model, and prior
preconditioner. Thus the framework is matrix-free at the level of
\(J^{\!\top}R^{-1}J\), but the generalized eigenproblem may still require
scalable prior or covariance actions.

\section{Examples and applications}
\label{sec:examples_applications}

The examples are chosen to expose different aspects of the same operator-level
assessment at increasing levels of mechanical and observational complexity. The simply
supported beam gives a closed-form kernel; the two-span beam shows support-induced
spatial coupling; the static--dynamic benchmark illustrates heterogeneous information
fusion; and the two-dimensional damage example uses leading information modes as a
reconstruction subspace.

\subsection{A simply supported beam under a moving point load}

To show that the framework can yield interpretable closed-form structure, consider a
simply supported Euler--Bernoulli beam of length \(L\) with a rotation sensor at
\(r \in (0,L)\). Each factor in the information operator has a mechanical meaning: the
primal bending-moment field describes the moving-load excitation, the adjoint influence
field describes the rotation measurement, and their product gives the local sensitivity to
compliance perturbations.
Let \(z\) denote the position of a moving point load \(P\). Let \(v(x)=1/EI(x)\) denote the flexural compliance field, and let \(\delta v(x)\) be a small compliance perturbation.
Introduce the normalized coordinates
\[
s=\frac{x}{L},\qquad
\bar s=\frac{\bar x}{L},\qquad
\zeta=\frac{z}{L},\qquad
\rho=\frac{r}{L}.
\]
The linearized rotation response takes the form
\begin{equation}
\begin{aligned}
\delta \theta(r;z)
&=
\int_0^L K(r,z;x)\,\delta v(x)\,dx,\\
K(r,z;x)
&=
\mu_\rho(s)\,M(x;z)
=
PL\,\mu_\rho(s)\mathcal M(s;\zeta),
\end{aligned}
\label{eq:beam_kernel}
\end{equation}
where \(M(x;z)\) is the bending-moment field generated by the moving point load and \(\mu_\rho(s)\) is the piecewise-affine adjoint bending-moment influence function associated with a unit positive rotation measurement at \(\rho\), up to the chosen rotation sign convention.

Equivalently, \(\mu_\rho\) is obtained by applying the rotation measurement functional as an adjoint load; the product form in \cref{eq:beam_kernel} follows from virtual work for a compliance perturbation.
For the simply supported beam,
\begin{equation}
\mu_\rho(s)=
\begin{cases}
-s, & s < \rho,\\[4pt]
1-s, & s > \rho,
\end{cases}
\qquad
\mathcal M(s;\zeta)=
\min\{s,\zeta\}\,[1-\max\{s,\zeta\}],
\label{eq:beam_influence_functions}
\end{equation}
with \(M(x;z)=PL\,\mathcal M(s;\zeta)\). The function \(\mathcal M\) is the normalized
triangular moving-load bending-moment influence line, equivalently the Green's function
of \(-d^2/ds^2\) on \([0,1]\) with homogeneous Dirichlet boundary conditions.

For a finite set of load positions \(z_k\), independent rotation measurements with noise
variance \(\sigma_k^2\) give the discrete information kernel
\begin{equation}
\mathcal I_r^v(x,\bar x)
=
\sum_{k}
\frac{1}{\sigma_k^2}
K(r,z_k;x)K(r,z_k;\bar x).
\label{eq:beam_discrete_information_kernel}
\end{equation}
The continuous sweep used below should be read as a weighted design-measure limit of
this finite observation model, not as an assumption of independent noise at
uncountably many load positions. Equivalently, a quadrature-normalized discrete model
would replace \(1/\sigma_k^2\) in \cref{eq:beam_discrete_information_kernel} by
\(w_k/\sigma_k^2\), giving a continuum weight \(d\mu_\sigma(z)\). The dense uniform
convention used in the closed-form formulas below corresponds to
\(d\mu_\sigma(z)=dz/\sigma^2\), so the associated continuous local information
operator on compliance perturbations has kernel

\begin{equation}
\mathcal I_r^v(x,\bar x)
=
\frac{1}{\sigma^2}\int_0^L K(r,z;x)\,K(r,z;\bar x)\,dz.
\label{eq:beam_information_kernel}
\end{equation}

\begin{proposition}[Continuous information kernel]
\label{prop:beam_full_kernel_closed_form}
Assume a dense uniform sweep of load positions over \(z \in [0,L]\). Let
\[
\alpha=\min(s,\bar s),
\qquad
\beta=\max(s,\bar s).
\]
Then the moment-product integral in \cref{eq:beam_information_kernel} is
\begin{equation}
\int_0^L M(x;z)M(\bar x;z)\,dz
=
\frac{P^2L^3}{6}\,
\alpha(1-\beta)\left(2\beta-\alpha^2-\beta^2\right),
\label{eq:beam_moment_product_integral}
\end{equation}
and the full Fisher-information kernel for compliance perturbations is
\begin{equation}
\mathcal I_r^v(x,\bar x)
=
\frac{P^2L^3}{6\sigma^2}\,
\mu_\rho(s)\mu_\rho(\bar s)\,
\alpha(1-\beta)\left(2\beta-\alpha^2-\beta^2\right).
\label{eq:beam_full_kernel_closed_form}
\end{equation}
\end{proposition}

\begin{proof}
Assume without loss of generality that \(s\le \bar s\), so that \(\alpha=s\) and
\(\beta=\bar s\). Splitting the normalized load-position integral at \(\zeta=s\) and
\(\zeta=\bar s\) gives
\begin{align}
\int_0^L M(x;z)M(\bar x;z)\,dz
&=
P^2L^3\bigg[
(1-s)(1-\bar s)
\int_0^s \zeta^2\,d\zeta \notag\\
&\qquad
+s(1-\bar s)
\int_s^{\bar s} \zeta(1-\zeta)\,d\zeta
+s\bar s\int_{\bar s}^1 (1-\zeta)^2 d\zeta
\bigg] \notag\\
&=
\frac{P^2L^3}{6}\,
\alpha(1-\beta)\left(2\beta-\alpha^2-\beta^2\right).
\end{align}
The result is symmetric in \(x\) and \(\bar x\). Substitution into
\cref{eq:beam_information_kernel} yields \cref{eq:beam_full_kernel_closed_form}.
\end{proof}

Its diagonal,
\begin{equation}
\mathcal I_{r,\mathrm{diag}}^v(x) := \mathcal I_r^v(x,x),
\label{eq:beam_diagonal_density_def}
\end{equation}
defines a pointwise information density for compliance perturbations in the chosen
coordinate representation; correlated spatial perturbation patterns are governed by the
full kernel \(\mathcal I_r^v(x,\bar x)\).

The kernel in \cref{eq:beam_full_kernel_closed_form} need not be pointwise nonnegative,
because \(\mu_\rho(s)\mu_\rho(\bar s)\) changes sign across the sensor location. This does not
contradict positive semidefiniteness: positivity applies to the quadratic form
\(\int\!\int \delta v(x)\mathcal I_r^v(x,\bar x)\delta v(\bar x)\,dx\,d\bar x\), not to
each kernel value. The normalized kernel in \cref{fig:simple_beam_kernel} makes this
distinction visible: the diagonal gives a local magnitude, whereas the off-diagonal
lobes and sign changes determine how separated stiffness perturbations reinforce or
cancel each other in the measured rotation response.

If the parameter of interest is flexural rigidity \(EI(x)\) rather than compliance, the chain rule gives
\[
\frac{\partial v}{\partial EI}=-\frac{1}{EI^2}=-v^2,
\]
and therefore
\begin{equation}
\mathcal I_r^{EI}(x,\bar x)
=
v(x)^2\,\mathcal I_r^v(x,\bar x)\,v(\bar x)^2.
\label{eq:beam_ei_kernel_transform}
\end{equation}
Thus the pointwise information density for \(EI\) identification is \(\mathcal I_{r,\mathrm{diag}}^{EI}(x)=v(x)^4\mathcal I_{r,\mathrm{diag}}^v(x)\), with the multiplicative factor \(v(x)^4=EI(x)^{-4}\).

\begin{proposition}[Diagonal information density]
\label{prop:beam_closed_form}
Assume a dense uniform sweep of load positions over \(z \in [0,L]\).
Then the diagonal of the continuous information kernel is
\begin{equation}
\mathcal I_{r,\mathrm{diag}}^v(x)
=
\frac{1}{\sigma^2}\int_0^L K(r,z;x)^2\,dz
=
\frac{P^2L^3}{3\sigma^2}\,\mu_\rho(s)^2\,s^2(1-s)^2,
\qquad 0<s<1.
\label{eq:beam_information_density}
\end{equation}
Equivalently, in normalized form,
\begin{equation}
\mathcal I_{\rho,\mathrm{diag}}^v(s)
=
\kappa_v\,s^2(1-s)^2\mu_\rho(s)^2,
\label{eq:beam_normalized_density}
\end{equation}
where \(\kappa_v>0\) absorbs physical and sampling-convention factors.
\end{proposition}

\begin{proof}
Using \(M(x;z)=PL\,\mathcal M(s;\zeta)\),
\begin{align}
\int_0^L M(x;z)^2\,dz
&=
P^2L^3\left[
(1-s)^2 \int_0^s \zeta^2\,d\zeta
+
s^2 \int_s^1 (1-\zeta)^2 d\zeta
\right] \notag\\
&=
\frac{P^2L^3}{3}s^2(1-s)^2.
\end{align}
Substituting into the definition of \(\mathcal I_{r,\mathrm{diag}}^v(x)\) yields \cref{eq:beam_information_density}.
The normalized form follows by absorbing the constant prefactor into \(\kappa_v\); its two
polynomial branches are implied by the two branches of \(\mu_\rho(s)\).
\end{proof}

\Cref{prop:beam_closed_form} explains several qualitative features that are easy to misread as numerical artifacts.
In the limiting case of a sensor at the support (\(\rho = 0\)), the interior maximum occurs at \(s = 1/3\).
In the limiting case of a sensor at the opposite support (\(\rho = 1\)), it occurs at \(s = 2/3\).
For an interior sensor, the pointwise information density has an analytic jump at the sensor location:
\begin{equation}
\frac{\mathcal I_{\rho,\mathrm{diag}}^v(\rho^+)}{\mathcal I_{\rho,\mathrm{diag}}^v(\rho^-)} = \frac{(1-\rho)^2}{\rho^2}.
\label{eq:beam_jump_ratio}
\end{equation}
The value at \(s=\rho\) is immaterial for the integral operator; the discontinuity refers
to one-sided limits of the ideal point-sensor model. Finite gauge length or regularized
observation operators would smooth this jump. The local visibility is therefore
asymmetric even in a simple benchmark, because the moving-load experiment and the rotation
influence field encode directional structure before any numerical discretization is
introduced.

Analytical state-observability-Gramian results for Euler--Bernoulli beams in
sensor-placement settings~\cite{braceSensorPlacementCantilever2022} corroborate the
intuition that sensor location strongly affects which patterns are visible. The present
density, however, is a parameter-level Fisher-information diagnostic rather than a
classical state-observability Gramian.

If the rotation-sensor weighting is removed, let \(K_0(x,\bar x)\) denote the kernel of the unweighted moment-product operator \(\mathcal K_0\):
\begin{equation}
K_0(x,\bar x)
=
\int_0^L M(x;z)M(\bar x;z)\,dz .
\label{eq:unweighted_moment_kernel}
\end{equation}

Equivalently, at the operator level in dimensional \(x\)-coordinates, let
\(A=-d^2/dx^2\) on \((0,L)\) with homogeneous Dirichlet boundary conditions, and let
\(\mathcal G_2=A^{-2}\) denote the unforced Green-product operator. Then
\begin{equation}
\mathcal K_0
=
\int_0^L M(\cdot;z)\otimes M(\cdot;z)\,dz
=
P^2\mathcal G_2
=
P^2 A^{-2}.
\label{eq:unweighted_moment_operator}
\end{equation}
Here \(f\otimes f\) denotes the rank-one operator
\((f\otimes f)g=f\langle f,g\rangle_{L^2(0,L)}\). The operator \(A\) acts in the
dimensional coordinate \(x\); using the normalized coordinate \(s=x/L\) gives the same
inverse-Laplacian structure with the corresponding powers of \(L\).

The corresponding eigenfunctions are the sine modes
\(\sin(n\pi s)\) in normalized coordinate, with eigenvalues proportional to
\((L/(n\pi))^4\), including the factor \(P^2/\sigma^2\) when the noise weighting is applied.
The actual one-sensor compliance-information operator is instead
\begin{equation}
\mathcal I_r^v
=
\frac{1}{\sigma^2}M_{\mu_\rho}\mathcal K_0M_{\mu_\rho}
=
\frac{P^2}{\sigma^2}\,M_{\mu_\rho}A^{-2}M_{\mu_\rho},
\label{eq:sensor_weighted_operator}
\end{equation}
where \(M_{\mu_\rho}\) denotes multiplication by \(\mu_\rho(s)\) after the coordinate
mapping \(x=Ls\). Thus \(A^{-2}\) denotes the unforced Green-product operator, while the
load factor \(P^2\) enters only once through \(\mathcal K_0\).
The resulting modes solve the compact self-adjoint Fredholm eigenproblem
\begin{equation}
\int_0^L \mathcal I_r^v(x,\bar x)\psi_i(\bar x)\,d\bar x
=
\lambda_i\psi_i(x),
\label{eq:beam_information_mode_fredholm}
\end{equation}
and, except in the unweighted case, they do not reduce to elementary sine functions.
The exact kernel nevertheless gives a semi-analytic benchmark.
With the normalized sine basis
\(\widehat\varphi_n(s)=\sqrt{2}\sin(n\pi s)\), the weighted operator has matrix entries
\begin{equation}
T_{ij}
=
\frac{P^2}{\sigma^2}
\sum_{n=1}^{\infty}
\left(\frac{L}{n\pi}\right)^4
B_{in}(\rho)B_{jn}(\rho),
\qquad
B_{in}(\rho)=\int_0^1 \widehat\varphi_i(s)\mu_\rho(s)\widehat\varphi_n(s)\,ds,
\label{eq:beam_sine_galerkin_modes}
\end{equation}
where \(B_{in}(\rho)\) is available in closed form because \(\mu_\rho\) is piecewise affine.
The modes plotted in \cref{fig:single_sensor_information_modes} are therefore
semi-analytic parameter-observability modes computed from the exact sensor-weighted
kernel. Their rapid spectral decay shows the effective low-dimensionality of this sensing
configuration, while the mode shapes show how the point-sensor influence line imprints a
localized change at the sensor position. In this and later numerical examples,
information modes are computed in the chosen discrete parameter coordinates and should be
interpreted as fixed-discretization diagnostics unless a metric- or prior-preconditioned
formulation is explicitly used.

\begin{figure}[!htbp]
\centering
\includegraphics[width=0.8\textwidth]{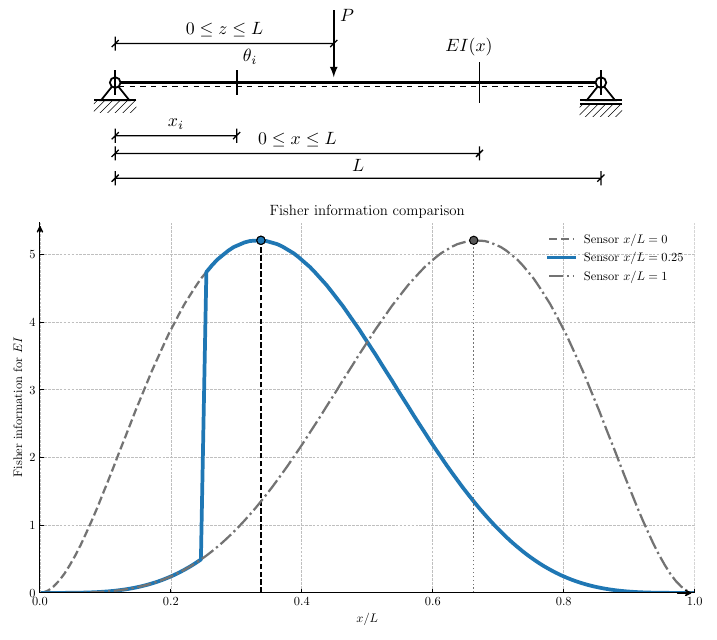}
\caption{Numerical rendering of the pointwise information density \(\mathcal I_{\rho,\mathrm{diag}}^v(s)\) from \cref{eq:beam_normalized_density} for a quarter-span rotation sensor (solid blue) and for support sensors (dashed gray). The jump discontinuity at \(s=\rho\) reflects the piecewise-affine influence field \(\mu_\rho\); finite sensor gauge length or regularized observation operators smooth this jump into a sharp transition. The right-shifted maximum for \(\rho=1/4\) illustrates how sensor position biases local visibility under a moving-load sweep.}
\label{fig:beam_density}
\end{figure}

\begin{figure}[!htbp]
\centering
\includegraphics[width=\textwidth]{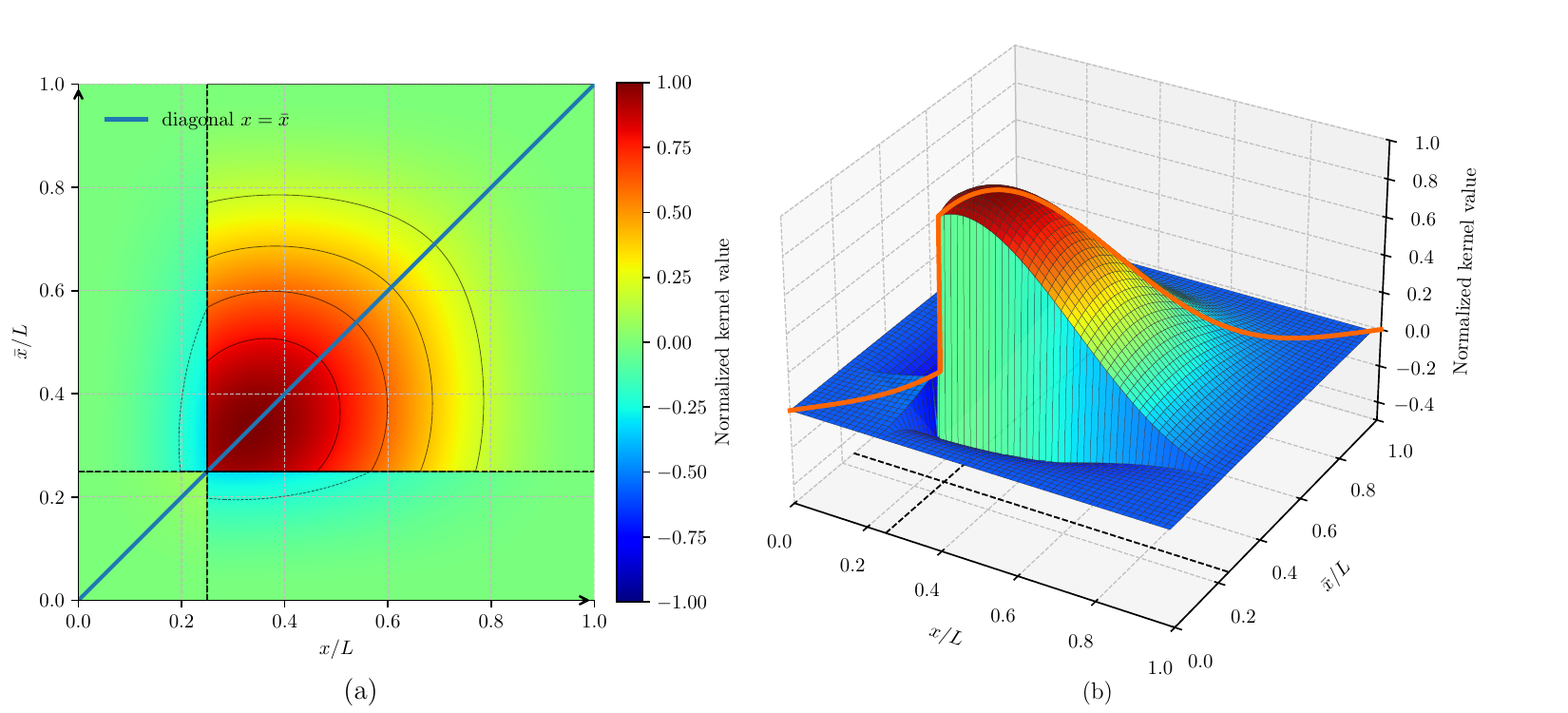}
\caption{Full local information kernel for a simply supported beam with a rotation sensor at \(\rho=0.25\). The plot shows the normalized \(EI\)-kernel \(\mathcal I_r^{EI}(x,\bar{x})\); for a uniform reference stiffness this has the same spatial pattern as the compliance kernel up to a positive constant factor. 
(a)~Contour plot of the normalized kernel with the diagonal \(x=\bar{x}\) (blue line), sensor position (dashed), and contour lines. 
(b)~Three-dimensional surface representation showing the kernel structure and the diagonal direction.
}
\label{fig:simple_beam_kernel}
\end{figure}

\begin{figure}[H]
\centering
\includegraphics[width=\textwidth]{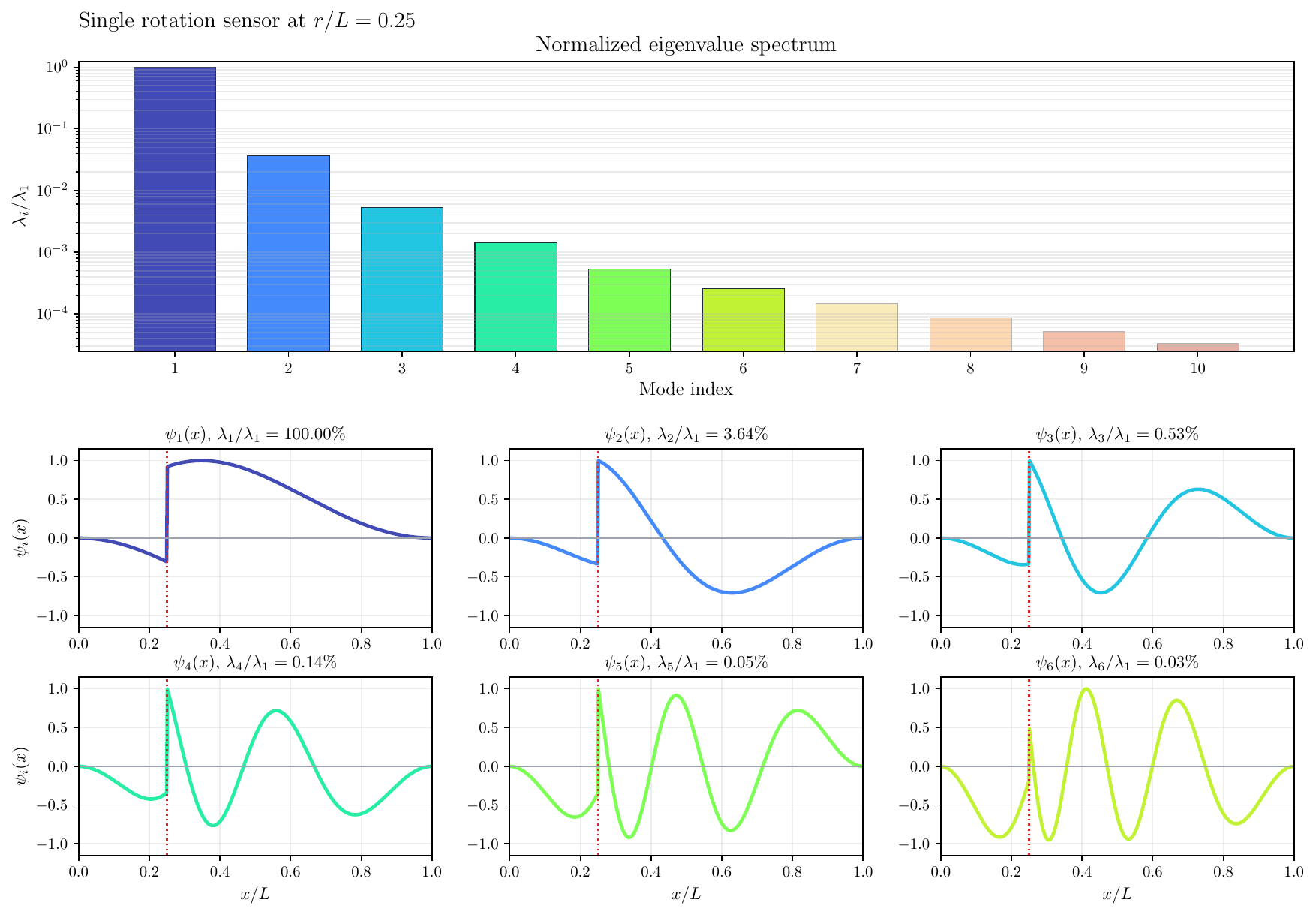}
\caption{Normalized spectrum and leading local parameter-observability modes for a single rotation sensor at \(\rho=0.25\), computed from the exact sensor-weighted information kernel. The normalized eigenvalue decay shows that most information is concentrated in a few parameter directions, while the plotted modes \(\psi_i(x)\) show increasingly weakly visible spatial patterns; the dashed vertical line marks the sensor position.}
\label{fig:single_sensor_information_modes}
\end{figure}

\FloatBarrier

\subsection{A two-span continuous beam under a moving point load}
\label{sec:two_span_beam}

As a second static example, consider a continuous Euler--Bernoulli beam with two equal
spans. Let \(L_s\) denote the span length, \(L_T=2L_s\) the total length, and let the
vertical supports be located at \(x=0\), \(x=L_s\), and \(x=2L_s\). For plotting, use the
span-normalized coordinate \(s=x/L_s\) and sensor coordinate \(\rho_s=r/L_s\). A rotation
sensor is placed in the first span, with the plotted full-kernel example using
\(\rho_s=0.25\), and the point load is swept over both spans. The linearized response has
the same product-kernel structure as \cref{eq:beam_kernel}, but over the two-span domain:
\begin{equation}
\delta\theta(r;z)
=
\int_0^{L_T}K(r,z;x)\,\delta v(x)\,dx,
\qquad
K(r,z;x)=\mu_r(x)M_z(x).
\label{eq:twospan_static_kernel}
\end{equation}
Here the influence fields \(M_z(x)\) and \(\mu_r(x)\) are those of the continuous beam rather than the simple triangular functions in \cref{eq:beam_influence_functions}. The continuous Fisher-information kernel is therefore
\begin{equation}
\mathcal I_r^v(x,\bar x)
=
\frac{1}{\sigma^2}
\int_0^{L_T}
\mu_r(x)M_z(x)\,
\mu_r(\bar x)M_z(\bar x)\,dz,
\qquad
\mathcal I_{r,\mathrm{diag}}^v(x)=\mathcal I_r^v(x,x).
\label{eq:twospan_continuous_kernel}
\end{equation}
The integral is interpreted with the same weighted design-measure convention as in the
simply supported beam example.

For \(EI(x)\) as the parameter, the compliance-kernel transformation in \cref{eq:beam_ei_kernel_transform} is applied pointwise in \(x\) and \(\bar x\).
In finite-element form, each moving-load position \(z_k\) defines a primal solve and the rotation sensor defines one adjoint solve,
\begin{equation}
\mathbf K_{\mathrm{FE}} u_{z_k}=f_{z_k},
\qquad
\mathbf K_{\mathrm{FE}}\lambda_r^{\mathrm{adj}}=\ell_r ,
\label{eq:multispan_primal_adjoint}
\end{equation}
where \(f_{z_k}\) is the nodal load vector and \(\ell_r\) represents the rotation measurement functional. With curvature operator \(B(x)\) and reference stiffness \(EI_0(x)\), the corresponding primal and adjoint bending-moment fields are
\begin{equation}
M_{z_k}(x)=EI_0(x)B(x)u_{z_k},
\qquad
\mu_r(x)=EI_0(x)B(x)\lambda_r^{\mathrm{adj}} .
\label{eq:multispan_moment_fields}
\end{equation}
Thus the local sensitivity for compliance is
\begin{equation}
K(r,z_k;x)=\mu_r(x)M_{z_k}(x),
\label{eq:multispan_sensitivity_kernel}
\end{equation}

Analogously to the discrete kernel in \cref{eq:beam_discrete_information_kernel}, the full Fisher-information kernel is assembled as
\begin{equation}
\mathcal I_r^v(x,\bar x)
\approx
\sum_{k=1}^{N_z}
\frac{w_k}{\sigma_k^2}\,
K(r,z_k;x)K(r,z_k;\bar x),
\qquad
\mathcal I_{r,\mathrm{diag}}^v(x)=\mathcal I_r^v(x,x),
\label{eq:multispan_discrete_kernel}
\end{equation}
where \(w_k\) are load-position quadrature weights.

\Cref{fig:two_span_fisher_information} shows the pointwise information density for the
equal-span continuous beam. The interior support redistributes bending moments, producing
sensor-dependent lobes and sharp changes in the sensitivity footprint rather than a
uniform coverage map. \Cref{fig:two_span_beam_kernel} shows the corresponding full kernel.
Its off-diagonal blocks diagnose whether stiffness perturbations in different regions or
spans produce correlated measurement signatures. Large-magnitude off-diagonal blocks
indicate strong spatial coupling, while weak coupling indicates better separation by the
sensor/load programme. Thus the kernel, not only its diagonal, is needed to assess span
coupling and pattern identifiability.

\begin{figure}[!htbp]
\centering
\includegraphics[width=0.8\textwidth]{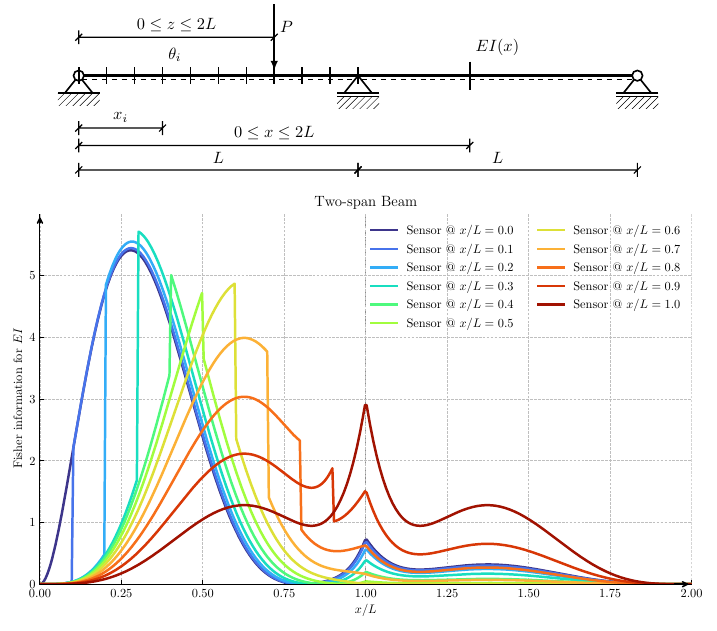}
\caption{Pointwise Fisher-information density \(\mathcal I_{r,\mathrm{diag}}^{EI}(x)\) for \(EI\) along a two-span continuous beam under moving-load excitation and rotation sensing. The colored curves show representative sensor positions along the two-span domain, the red envelope summarizes the maximum visibility obtained over the moving sensor sweep, and the vertical dashed line marks the interior support at \(s=1\).}
\label{fig:two_span_fisher_information}
\end{figure}

\begin{figure}[!htbp]
\centering
\includegraphics[width=\textwidth]{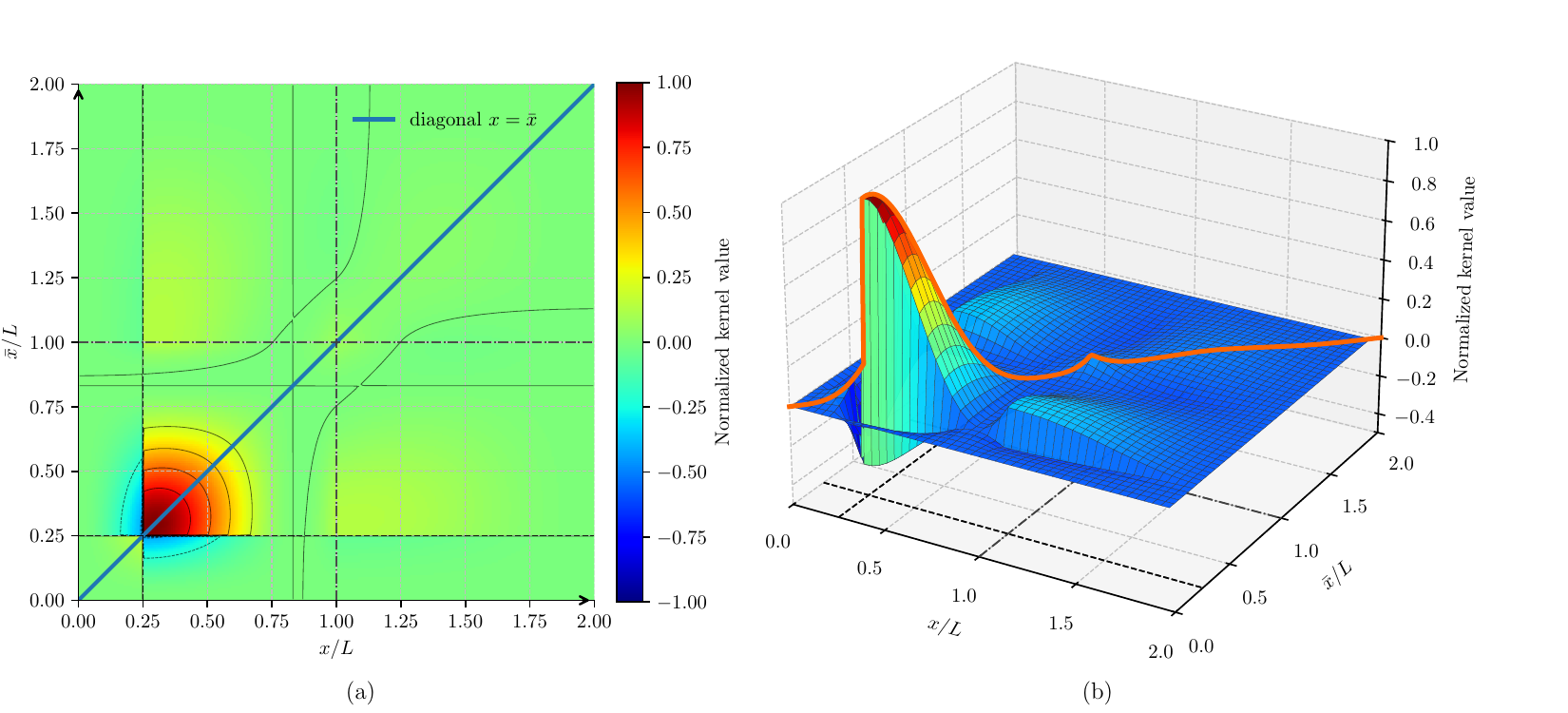}
\caption{Full local information kernel \(\mathcal I_r^{EI}(x,\bar{x})\) for a two-span continuous beam with a rotation sensor at \(\rho_s=0.25\) and an interior support at \(s=1\). 
(a)~Contour plot of the normalized kernel with the diagonal \(x=\bar{x}\) (blue line), sensor position (dashed), interior support (dash-dotted), and contour lines. 
(b)~Three-dimensional surface representation showing the kernel structure and the influence of the interior support.}
\label{fig:two_span_beam_kernel}
\end{figure}

\FloatBarrier

\subsection{Beam \texorpdfstring{\(EI\)}{EI} identification benchmark}
\label{sec:beam_ei_identification_benchmark}

The preceding examples diagnosed local visibility for idealized beam observations. We now
consider an inverse-identification benchmark for a damaged simply supported
Euler--Bernoulli beam. The unknown field is the spatially varying flexural rigidity
\(EI(x)\). A localized stiffness reduction is introduced in the right half of the span,
and the goal is to identify this damaged region from static, dynamic, and fused
static--dynamic observations.

The parameter field is represented in elementwise log-stiffness coordinates,
\begin{equation}
p_e=\log(EI_e/EI_0),
\label{eq:beam_log_stiffness_parameter}
\end{equation}
so that positivity of \(EI\) is enforced during inversion. The corresponding finite-dimensional Jacobians are taken with respect to \(p\). Thus the diagonal Fisher information shown below is an information density in log-stiffness coordinates, \(\mathcal I_{\mathrm{diag}}^{p}(x)\), rather than directly in \(EI(x)\).

The same damaged beam is interrogated by two observation systems: a static moving-load
test with two rotation sensors and a dynamic frequency-response test with several
harmonic excitation locations and one response sensor. The experiment configuration, data,
reconstructions, and information densities are summarized in
\cref{fig:hybrid_static_dynamic_composite}. This benchmark illustrates that heterogeneous
data types can be compared by the parameter-space directions they inform.
Static--dynamic data combination is established in structural model updating; here it is
used to show how the full information operator separates reinforcement of already visible
directions from genuine weak-direction gain~\cite{bellMultiresponseParameterEstimation2007,
kimSequentialFrameworkImproving2019}.

\subsubsection{Static tilt observations}

The static benchmark uses a moving vertical point load and two fixed tilt sensors. For
each load position \(z_k\), rotations are measured at sensor locations \(r_j\). The static
observation vector is obtained by stacking the measured rotations over all load positions
and sensor locations,
\begin{equation}
y_{\mathrm{stat}}
=
\{\theta(r_j;z_k)\}_{j,k}.
\label{eq:static_tilt_observation_vector}
\end{equation}
After linearization around a reference stiffness field, these data contribute a block
information operator of the form \cref{eq:sensor_information}, denoted here by
\(\mathcal I_{\mathrm{stat}}\). This is the finite-dimensional counterpart of the
moving-load rotation kernel derived above.

The static test is shown in \cref{fig:hybrid_static_dynamic_composite}a,b,d,g. The
identified static model follows the generated tilt responses, and the static-only
reconstruction recovers the broad location of the damaged region. The estimated stiffness
profile is smoother and less localized than the true loss, consistent with the integral
character of rotation measurements: each load case produces a broad sensitivity pattern.
The diagonal Fisher information is nonuniform and largest where the moving-load response
and rotation sensors have strong combined sensitivity, but it does not fully localize the
damaged interval; this explains the wider local posterior bands.

\subsubsection{Dynamic frequency-response observations}

Dynamic testing provides a different observation mechanism. Instead of measuring static
rotations under a moving load, harmonic or broadband excitation is applied and
frequency-response functions (FRFs) are measured. Such FRF-based and resonance-based
approaches are widely used in structural health monitoring and stiffness or
flexural-rigidity identification~\cite{bagheriIdentificationFlexuralRigidity2018,
kooSubstructuralIdentificationFlexural2015,leeFrequencydomainMethodStructural2002,
panResonancebasedApproachSection2021,bonkowskiStiffnessIdentificationReinforced2024}.
In the present benchmark, the dynamic observation vector stacks log-magnitude
frequency-response data over excitation positions and sampled frequencies,
\begin{equation}
y_{\mathrm{dyn}}
=
\{\log |\mathcal H(r,z_k;\omega_q)|\}_{k,q}.
\label{eq:dynamic_frf_observation_vector}
\end{equation}
The corresponding block information operator has the same form as
\cref{eq:sensor_information} and is denoted by \(\mathcal I_{\mathrm{dyn}}\). For
interpretation, consider harmonic point forcing
\[
f(x,t)=\hat F\,e^{i\omega t}\,\delta(x-z)
\]
applied at position \(z\) on a uniform simply supported beam. The steady-state
displacement \(w(x;\omega)\) satisfies
\begin{equation}
EI_0\,\frac{\partial^4 w}{\partial x^4}
-
\varrho A\,\omega^2 w
+
i\omega c_{\mathrm d}\,w
=
\hat F\,\delta(x-z),
\label{eq:forced_beam}
\end{equation}
where \(\varrho A\) is mass per unit length and \(c_{\mathrm d}\) is viscous damping.

For a generic evaluation coordinate \(a\), define \(\chi=a/L\).
Using normalized positions \(s=x/L\), \(\rho=r/L\), and \(\zeta=z/L\), the modal expansion of the Green's function is

\begin{equation}
G_0(a,x;\omega)
=
\frac{2}{L}
\sum_{n=1}^{\infty}
\frac{
\sin(n\pi\chi)
\sin(n\pi s)
}
{D_n(\omega)},
\qquad
D_n(\omega)
=
EI_0\Bigl(\frac{n\pi}{L}\Bigr)^{\!4}
-
\varrho A\,\omega^2
+
i\omega c_{\mathrm d},
\qquad n\in\mathbb N,
\label{eq:beam_greens_function}
\end{equation}
The curvature product that enters the local complex-FRF sensitivity at position \(x\),
for a sensor at \(r\) and excitation at \(z\), is
\begin{align}
&\frac{\partial^2 G_0}{\partial x^2}(r,x;\omega)\,
\frac{\partial^2 G_0}{\partial x^2}(x,z;\omega)
\notag\\
&\quad =
\frac{4}{L^2}
\sum_{j=1}^{\infty}
\sum_{n=1}^{\infty}
\frac{
\bigl(\frac{j\pi}{L}\bigr)^2
\sin(j\pi\rho)
\sin(j\pi s)
}
{D_j(\omega)}
\notag\\
&\qquad \times
\frac{
\bigl(\frac{n\pi}{L}\bigr)^2
\sin(n\pi\zeta)
\sin(n\pi s)
}
{D_n(\omega)} .
\label{eq:dynamic_curvature_sensitivity}
\end{align}
The benchmark, however, uses log-magnitude FRF data. Therefore the relevant Jacobian row
is not the complex sensitivity \(\partial \mathcal H/\partial p(x)\) itself, but the sensitivity of
\(\log |\mathcal H|\). For the log-stiffness coordinate
\(p(x)=\log(EI(x)/EI_0)\),
\begin{equation}
\frac{\partial \log |\mathcal H_{kq}|}{\partial p(x)}
=
\Re\!\left[
\frac{1}{\mathcal H_{kq}}
\frac{\partial \mathcal H_{kq}}{\partial p(x)}
\right],
\qquad
\mathcal H_{kq}=\mathcal H(r,z_k;\omega_q).
\label{eq:log_frf_sensitivity}
\end{equation}
For a set of excitation positions \(\{z_k\}_{k=1}^{K}\) and frequencies
\(\{\omega_q\}_{q=1}^{Q}\), the corresponding pointwise Fisher-information density in
log-stiffness coordinates is
\begin{equation}
\mathcal I_{\mathrm{dyn},\mathrm{diag}}^{p}(x;r)
=
\sum_{q=1}^{Q}
\sum_{k=1}^{K}
\frac{1}{\sigma_{\log,kq}^{2}}
\left[
\Re\!\left\{
\frac{1}{\mathcal H(r,z_k;\omega_q)}
\frac{\partial \mathcal H(r,z_k;\omega_q)}{\partial p(x)}
\right\}
\right]^2 .
\label{eq:dynamic_information_density}
\end{equation}
For a uniform reference beam, \(\partial \mathcal H/\partial p(x)\) is proportional, up to sign and
reference-stiffness factors, to
\[
\hat F(\omega_q)
\frac{\partial^2 G_0}{\partial x^2}(r,x;\omega_q)
\frac{\partial^2 G_0}{\partial x^2}(z_k,x;\omega_q),
\]
so \cref{eq:dynamic_curvature_sensitivity} supplies the modal series entering the
log-magnitude sensitivity. Unlike the static moving-load case, the general multi-frequency
expression does not reduce to a simple polynomial form. In a single-mode-dominated regime,
with light damping and excitation near \(\omega_{\ell}\), and away from antiresonance zeros
of \(\mathcal H\), the leading spatial factor in the density is expected to scale like
\begin{equation}
\mathcal I_{\mathrm{dyn},\mathrm{diag}}^{p}(x;r)
\approx
\kappa_{\mathrm{dyn}}\sin^4(\ell\pi s),
\label{eq:single_mode_dynamic}
\end{equation}
which vanishes at the supports and peaks near the curvature maxima of mode \(\ell\).

The numerical FRF benchmark is shown in
\cref{fig:hybrid_static_dynamic_composite}a,c,e,g. The frequency-response curves show
resonance and antiresonance structure that depends on excitation location, producing rows
of \(J_{\mathrm{dyn}}\) that differ qualitatively from the static rotation rows.
Consequently, the dynamic Fisher information is more concentrated around the damaged
interval than the static contribution. The dynamic-only reconstruction gives sharper
localization than the static-only reconstruction, although uncertainty remains near
boundaries and in regions weakly excited by the selected modes and sensor layout.

\subsubsection{Hybrid static--dynamic fusion}

The static and dynamic tests provide complementary information about \(EI(x)\). The static
moving-load test produces broad integral sensitivity patterns associated with rotation
influence lines, whereas the dynamic FRF test produces frequency-dependent patterns
associated with modal curvature, resonance amplification, and antiresonance structure.
Within the information-operator framework, the two data sources are fused by stacking
their residuals and Jacobians.

Under conditional independence of the static and dynamic measurement errors, the additivity rule in \cref{eq:additivity} gives
\begin{equation}
\mathcal I_{\mathrm{hyb}}
=
\mathcal I_{\mathrm{stat}}
+
\mathcal I_{\mathrm{dyn}} .
\label{eq:hybrid_information_additivity}
\end{equation}
In the numerical implementation, each residual block is normalized by its prescribed
noise standard deviation before forming the Jacobian. Thus the covariance weighting is
already included in the whitened Jacobians, and the static, dynamic, and hybrid
information matrices are the corresponding whitened Gram matrices, combined according to
\cref{eq:hybrid_information_additivity}. The same Gaussian prior/regularization
precision is used for the static-only, dynamic-only, and hybrid cases,
\begin{equation}
Q_{\mathrm{pr}}
=
\gamma_{\mathrm{pr}}D_{\mathrm{diff}}^{\!\top}D_{\mathrm{diff}}+\epsilon I ,
\label{eq:hybrid_prior_precision}
\end{equation}
where \(D_{\mathrm{diff}}\) is a first-difference operator on the log-stiffness parameters
and \(\gamma_{\mathrm{pr}}>0\) is its regularization weight. Using the local
posterior-precision structure in \cref{eq:local_post_precision}, with \(J\) denoting the
noise-whitened Jacobian for the corresponding data configuration, the local posterior
covariance is approximated by
\begin{equation}
C_{\mathrm{post}}
\approx
\left(J^{\!\top}J+Q_{\mathrm{pr}}\right)^{-1}.
\label{eq:hybrid_local_covariance}
\end{equation}

The uncertainty bands in \cref{fig:hybrid_static_dynamic_composite}d--f are local
linearized posterior bands obtained from this covariance approximation using the same
prior precision and known noise model for the three data configurations.

The fused result is shown in \cref{fig:hybrid_static_dynamic_composite}f,g. The hybrid
diagonal Fisher-information density is the pointwise sum of the static and dynamic
diagonal contributions and is largest near the damaged interval. Correspondingly, the
hybrid reconstruction has a deeper and more localized stiffness dip than the static-only
result and narrower local posterior bands. The density also shows where one block supplies
local visibility in regions where the other is weaker. This is a coordinate-level analogue
of weak-direction gain; a full weak-direction analysis would use the complete operator and
the projector \(\Pi_{\mathrm{weak}}\).

The practical trade-off remains important. Dynamic testing requires controlled
excitation, frequency selection, damping assumptions, and careful amplitude limits, while
static moving-load tests may be cheaper and less sensitive to damping and mass modelling.
The information-operator framework does not remove this engineering trade-off; it makes
the incremental information gain comparable to the marginal cost and complexity of the
additional test.

\begin{figure}[!htbp]
\centering
\includegraphics[width=\textwidth]{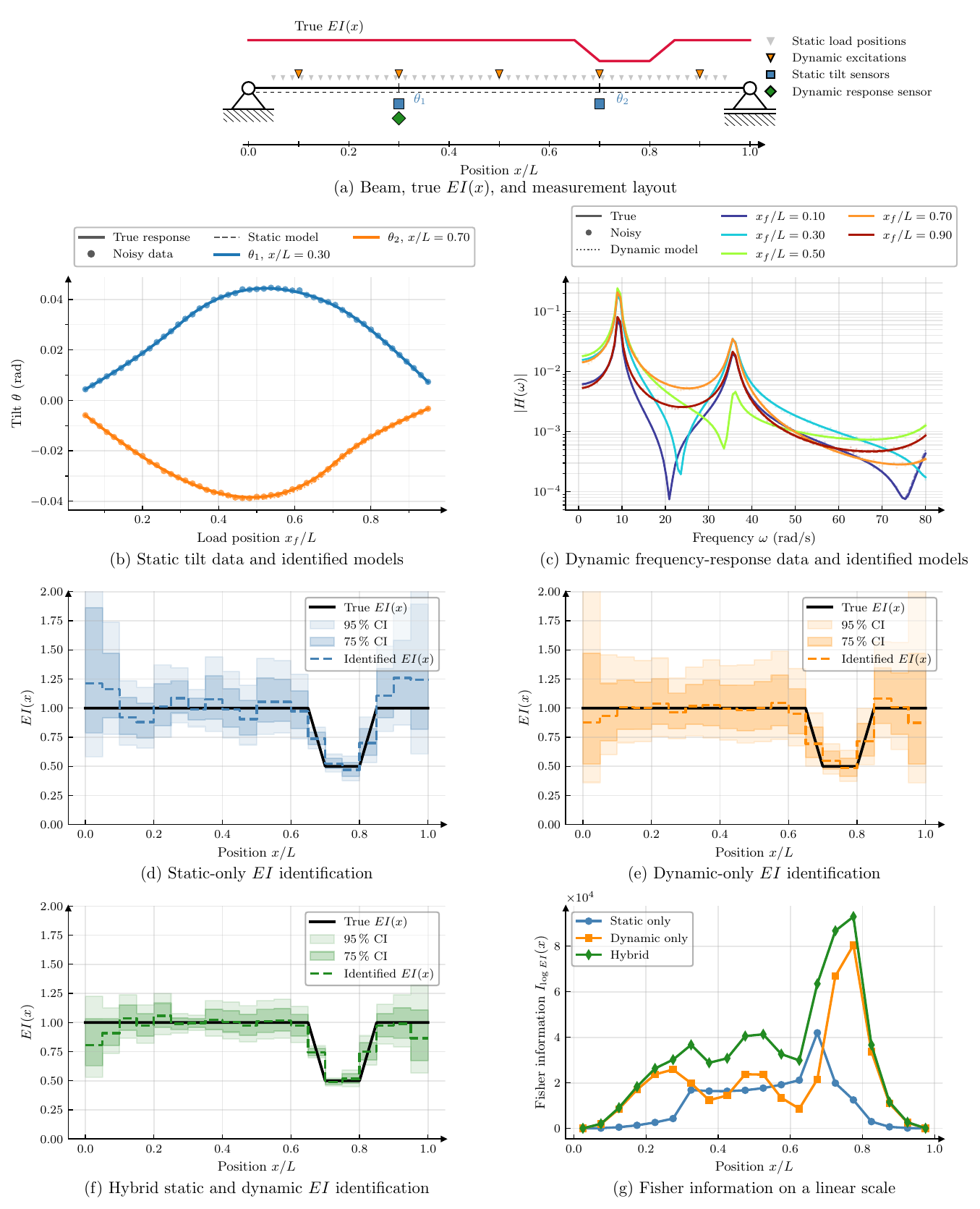}
\caption{Hybrid static--dynamic identification of a damaged flexural-rigidity field \(EI(x)\).
(a) Beam, true stiffness profile, static moving-load positions, dynamic excitation locations, and sensor layout.
(b) Static tilt responses under the moving load and model responses from the identified fields.
(c) Dynamic frequency-response functions for several excitation locations and corresponding identified-model responses.
(d)--(f) Static-only, dynamic-only, and hybrid reconstructions of \(EI(x)\).
Bands denote local linearized posterior uncertainty bands evaluated using the same prior precision and known noise model; dashed curves denote the corresponding MAP reconstructions.
(g) Diagonal Fisher-information densities in log-stiffness coordinates, shown on a linear scale. The hybrid density is the sum of the static and dynamic densities, illustrating information additivity under independent observation errors.}
\label{fig:hybrid_static_dynamic_composite}
\end{figure}

\subsection{Application to two-dimensional damage-field reconstruction}
\label{sec:2d_damage_reconstruction}

The final example applies the framework to a two-dimensional continuum
damage-identification problem. The structure is modelled as a beam-like linear-elastic
domain, but the unknown is a spatial damage field \(d(x,y)\) rather than a
one-dimensional beam stiffness. The leading information modes define the reconstruction
space, restricting the inverse update to damage patterns locally visible under the chosen
loads and sensors. This gives a finite-dimensional demonstration of the heterogeneous
observation map and operator viewpoint developed in \cref{sec:framework,sec:matrix_free}.

The unknown parameter is \(m=d(x,y)\), with the local elastic modulus represented as
\begin{equation}
E(x,y)=E_0\{\kappa+(1-\kappa)[1-d(x,y)]\}.
\label{eq:damage_to_modulus_2d}
\end{equation}
Thus \(d=0\) corresponds to the undamaged modulus and increasing \(d\) represents a local
stiffness loss, while the floor \(\kappa\) prevents a singular stiffness field.

For a prescribed experiment set \(\Xi\) of moving top-traction load cases, the
parameter-to-observation map has the same chain-rule structure as in
\cref{eq:chain_rule_factorization}. The physics solve maps \(m=d(x,y)\) to the
displacement and strain fields, and the observation operator then samples axial strain,
vertical displacement, and rotation-like quantities \(\partial u_y/\partial x\). Stacking
all load cases and sensor types gives the same stacked map \(Y(m;\Xi)\) used in
\cref{eq:stacked_map}.

Around a nominal damage field \(m_0\), the local Jacobian is
\(J(m_0;\Xi)=DY(m_0;\Xi)\), and the local information operator is the one defined in
\cref{eq:local_information_operator}. In the computation below, \(m_0\) is the nominal
mean damage field used for the benchmark and the center of the quadratic prior penalty,
and \(J\) is evaluated by finite differences of the physics model. The covariance \(R\)
is diagonal and encodes the relative reliability
of the strain, displacement, rotation-like, and model-error contributions.

If the observation blocks are conditionally independent, their contributions add exactly
as in \cref{eq:additivity}; the reconstruction in \cref{fig:2d_damage_reconstruction} is
therefore based on one heterogeneous local information operator. In contrast to the
analytic beam example, no closed-form kernel is used. The eigenvectors \(\psi_i\) rank
two-dimensional damage patterns by the noise-weighted output energy generated under the
selected load and sensor layout. Here they are likelihood-only Euclidean modes:
fixed-discretization diagnostics for the stated grid, parameter scaling, and observation
model.
A prior-relative Bayesian variant would instead use the generalized modes introduced in
\cref{eq:generalized_eigenproblem}.

The reconstruction step is the local Gaussian MAP update obtained from
\cref{eq:negative_log_posterior,eq:linearization}, restricted to the informed subspace.

Let \(\mathcal V_k=\operatorname{span}\{\psi_1,\ldots,\psi_k\}\). Only perturbations in this
subspace are allowed to change:

\begin{equation}
\begin{aligned}
m_{\mathrm{MAP},k}
&=
\operatorname*{arg\,min}_{m-m_0\in\mathcal V_k}
\Bigg[
\frac{1}{2}\left\|y-Y(m_0;\Xi)-J(m_0;\Xi)(m-m_0)\right\|_{R^{-1}}^2\\
&\qquad\qquad\qquad\qquad
+
\frac{1}{2}\|m-m_0\|_{Q_{\mathrm{pr}}}^2
\Bigg].
\end{aligned}
\label{eq:damage_subspace_reconstruction}
\end{equation}

Because the forward map is nonlinear, the informed subspace depends in principle on the
evaluation point. The fixed-subspace reconstruction in
\cref{eq:damage_subspace_reconstruction} is therefore a first-order dimension reduction
around \(m_0\); in strongly nonlinear regimes, the subspace would need to be updated
iteratively.
In the numerical example, the prior term is represented by a small isotropic penalty in
the retained modal coefficients, and the reconstructed damage field is projected to the
admissible interval \(0\le d\le0.9\). The information operator first ranks locally visible
spatial patterns and then defines the low-dimensional inverse-update space. This
moderate-size example explicitly forms the local information operator in
\cref{eq:local_information_operator}; large-scale implementations can instead use the
matrix-free actions described in \cref{sec:matrix_free}.

\Cref{fig:2d_damage_reconstruction} shows a representative reconstruction on a
\(17\times81\) grid, corresponding to \(1377\) damage unknowns. The observation vector
stacks \(14\) axial-strain sensors, \(3\) vertical-displacement sensors, and \(3\)
rotation-like sensors over eight moving-load cases, giving \(160\) scalar observations.
The reconstruction uses the first \(k=8\) Euclidean eigenmodes of the discrete local
information operator. This compact illustrative subspace retains the dominant visible
patterns while avoiding many weak, noise-sensitive modes. For this realization,
\(\lambda_1/\lambda_8 \approx 24.4\), indicating noticeably different information levels
within the retained subspace.

The top panel gives the true damage field, the moving load, and the sensor layout; the
bottom panel gives the MAP field reconstructed in the leading information subspace. The
damage is placed in the lower tension zone, where bending-induced axial strain gives
strong sensitivity to stiffness loss. The reconstruction recovers the dominant
longitudinal location and broad extent of the damaged region, while remaining smoother
than the true field and showing small localized artefacts near highly instrumented
regions.

This behavior is consistent with the spectrum of \(\mathcal I(m_0;\Xi)\): components in
the leading parameter-observability modes are preferentially reconstructed, whereas
components outside that subspace are suppressed or confounded by noise, finite sensor
coverage, and regularization. The reconstruction quality should therefore be interpreted
relative to the stated load cases, sensor types, noise model, and \(k=8\) information
subspace.

\begin{figure}[H]
\centering
\includegraphics[width=1.0\textwidth]{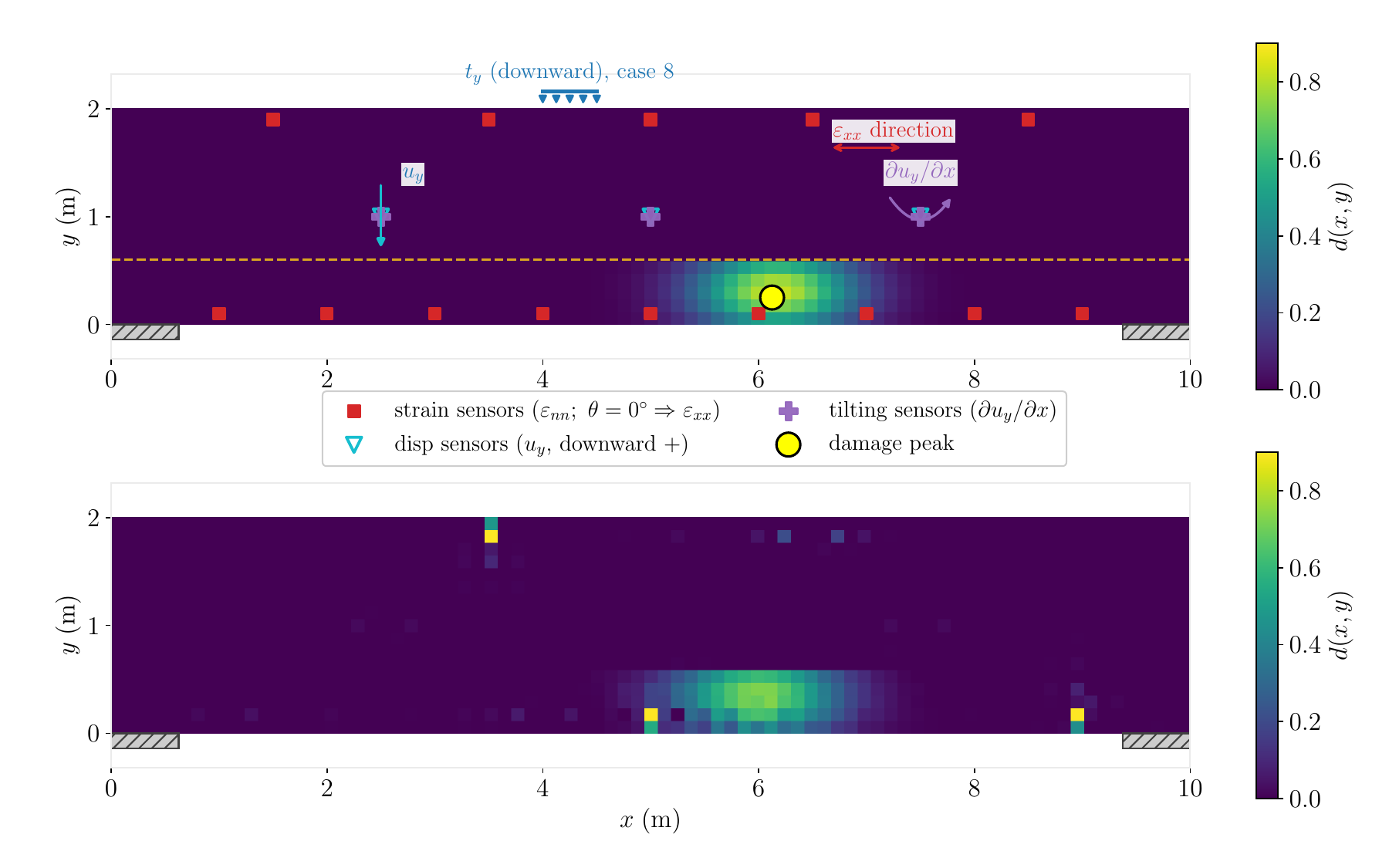}
\caption{Application of the parameter-observability framework to two-dimensional damage-field reconstruction. 
Top: true damage field \(d(x,y)\), moving downward load case, and sensor layout. The dashed horizontal line marks the lower tension-zone band; square markers denote axial strain sensors, inverted-triangle markers denote vertical displacement sensors, cross markers denote rotation/tilt sensors, and the circled marker denotes the damage peak. 
Bottom: MAP reconstruction restricted to the first \(k=8\) fixed-discretization parameter-observability modes of the local information operator.
The reconstruction localizes the dominant damage zone but remains smoother than the true field, illustrating how the local information operator selects the damage components most strongly supported by the chosen heterogeneous experiment.}
\label{fig:2d_damage_reconstruction}
\end{figure}

\section{Discussion and scope}
\label{sec:discussion}

The examples illustrate the main viewpoint of the paper: within the local linearized
setting, spatial identifiability is governed by the full local information operator,
not by its diagonal alone.
The bilinear form
defined by \(\mathcal I\) diagnoses visible, weakly visible, and locally invisible
parameter-field perturbations, while the kernel and spectrum expose coupling between
those perturbations. The diagonal remains useful as a pointwise visibility density, but
it cannot by itself rank coupled spatial patterns or distinguish redundancy from
weak-direction gain.

This operator provides a common language for concepts that are often treated separately:
Fisher information in experimental design, Gauss--Newton curvature in inverse problems,
and noise-weighted parameter-output sensitivity Gramians in sensitivity analysis.
In this framework, scalar criteria, diagonal densities, prior-relative modes, and
weak-direction assessments are not independent diagnostics, but different reductions or
projections of the same parameter-space information geometry
\cite{sanayeiSensorPlacementParameter2002,papadimitriouOptimalSensorPlacement2004,
vandewouwerApproachSelectionOptimal2000,ebrahimianInformationTheoreticApproachIdentifiability2019,
alexanderianOptimalExperimentalDesign2021,jauffretObservabilityFisherInformation2007a,
qianModelReductionLinear2022,cuiLikelihoodinformedDimensionReduction2014,
spantiniOptimalLowrankApproximations2015,scheffelsLikelihoodinformedModelReduction2025}.

This viewpoint also clarifies related design and adaptation methods. Diagonal
information densities have been used for mesh adaptation and
Fisher-information-based design diagnostics~\cite{bangerthEstimatingUsingInformation2026,
antilFisherInformationBasedSensorPlacement2026}; here, the diagonal is a visibility
density, while spatial identifiability is assessed through the full operator. This
distinction is consistent with classical inverse-problem resolution and covariance
analysis, where diagonal uncertainty summaries can miss correlations and spatial
spreading encoded by the full operator~\cite{menkeResolutionCovarianceGLS2014}.
Bayesian optimal experimental design supplies scalar posterior-uncertainty criteria for
ranking configurations~\cite{alexanderianOptimalExperimentalDesign2021,huanOptimalExperimentalDesign2024};
the information-operator view identifies the parameter-space directions responsible for
such gains.

The analysis is local: evaluating the Jacobian at a nominal field \(m_0\) gives a
first-order description of distinguishability around a baseline model, design state,
current digital-twin state, or maximum a posteriori (MAP) estimate. The choice of
\(m_0\) should follow the purpose of the analysis: for pre-experiment sensor placement it
is naturally the expected healthy or design state, whereas for post-experiment
interpretation it should usually be the MAP estimate or the current updated model state.
In damage-identification problems, a useful two-stage strategy is to compute modes at
the healthy baseline for experiment design and then recompute them at the MAP estimate.
Robustness can be assessed by comparing dominant subspaces at several plausible nominal
fields, for example through principal angles.

The useful range of the local operator is limited by the validity of the first-order
linearization around \(m_0\). For strongly nonlinear damage regimes, contact changes,
cracking, large stiffness losses, or changing boundary conditions, the identifiable
subspace can change substantially and should be updated along an optimization,
filtering, or sequential digital-twin path. More generally, the distinction between
structural and practical identifiability remains
essential~\cite{wielandStructuralPracticalIdentifiability2021}: the proposed operator
diagnoses local first-order distinguishability, not global uniqueness. Nonlinear
parameter interactions, multimodal posteriors, and practical non-identifiability away
from \(m_0\) require complementary tools such as nonlinear inversion, profile likelihood,
posterior sampling, continuation studies, or multi-start optimization.

The specific form in \cref{eq:local_information_operator} assumes additive Gaussian
observational uncertainty with parameter-independent covariance. For heavy-tailed,
outlier-contaminated, drifting, or environmentally dependent data, the corresponding
local operator should be taken as the Fisher information or Hessian of the chosen
likelihood, possibly with robust or residual-dependent
weights~\cite{huberRobustEstimationLocation1964a,huberRobustStatistics2009}. Correlated
observation errors can change information rankings and must be represented in
the joint covariance rather than treated as independent blocks~\cite{attiaOptimalExperimentalDesign2022,
qiOptimalExperimentalDesign2026}.
Similarly, covariance inflation for model discrepancy accounts for additional uncertainty
but does not remove systematic bias. If loads, boundary conditions, environmental effects, or
calibration parameters are confounded with the parameter field of interest, they should
be represented explicitly as nuisance parameters and eliminated or marginalized in the
joint local information geometry.

Finally, mode interpretation is metric dependent. Raw spectra change under
reparameterization and may depend on mesh resolution, basis scaling, and the chosen
inner product. Metric-weighted and prior-preconditioned modes are therefore part of the
framework. In realistic distributed inverse problems, the operator is often low-rank or
severely ill-conditioned, so prior regularization, Bayesian formulation, and subspace
truncation are part of the well-posed inverse problem rather than numerical
afterthoughts. Matrix-free actions make the framework scalable in principle, but
dominant-mode extraction still requires repeated tangent, covariance-weighting, prior,
and adjoint actions.

\section{Conclusions}
\label{sec:conclusions}

This paper has presented a local information-operator framework for assessing spatial
identifiability in distributed-parameter Bayesian inverse problems in computational
mechanics. Under Gaussian observational uncertainty and local linearization, the
likelihood contribution to posterior precision has the same quadratic form as Fisher
information, Gauss--Newton data-misfit curvature, and a noise-weighted parameter-output
sensitivity Gramian. The paper uses this equivalence to interpret local identifiability
as a property of parameter-field perturbations under a specified observation programme.

The resulting operator separates coordinate-level visibility from pattern-level
identifiability. The diagonal gives a pointwise information density, whereas the full
kernel and metric- or prior-preconditioned spectra rank the locally distinguishable
parameter-field patterns and reveal their coupling. Blockwise projections then
distinguish redundant information gain from genuine weak-direction gain when
heterogeneous observation blocks are added.

The beam and damage-identification examples illustrate this interpretation across
closed-form kernel analysis, finite-element information construction, static--dynamic
sensor fusion, and subspace-restricted reconstruction. The matrix-free formulation
further shows how the same diagnostics can be evaluated through tangent,
covariance-weighting, and adjoint actions, supporting their use as a diagnostic layer for
large-scale mechanics models. Across these cases, spatial identifiability is governed
less by the magnitude of pointwise information density than by the parameter-field
patterns distinguishable by the experiment. The main contribution is to make this
distinction explicit for distributed-parameter mechanics problems and to show how it
affects interpretation, sensor fusion, and reduced reconstruction.

Future work may extend the framework to large three-dimensional finite-element
digital twins, correlated and robust noise models, and adaptive experiment design with
sequentially updated information operators.

\printbibliography

@book{beckParameterEstimationEngineering1977,
  title = {Parameter Estimation in Engineering and Science},
  author = {Beck, James Vere and Arnold, Kenneth J.},
  year = {1977},
  publisher = {Wiley},
  address = {New York}
}

@book{tarantolaInverseProblemTheory2005,
  title = {Inverse Problem Theory and Methods for Model Parameter Estimation},
  author = {Tarantola, Albert},
  year = {2005},
  publisher = {Society for Industrial and Applied Mathematics},
  address = {Philadelphia, PA}
}

@article{papadimitriouEntropyBasedOptimalSensor2000,
  title = {Entropy-Based Optimal Sensor Location for Structural Model Updating},
  author = {Papadimitriou, Costas and Beck, James L. and Au, Siu-Kui},
  year = {2000},
  journal = {Journal of Vibration and Control},
  volume = {6},
  number = {5},
  pages = {781--800},
  doi = {10.1177/107754630000600508}
}

@article{papadimitriouOptimalSensorPlacement2004,
  title = {Optimal Sensor Placement Methodology for Parametric Identification of Structural Systems},
  author = {Papadimitriou, Costas},
  year = {2004},
  journal = {Journal of Sound and Vibration},
  volume = {278},
  number = {4},
  pages = {923--947},
  doi = {10.1016/j.jsv.2003.10.063}
}

@misc{lohnerHighFidelityDigitalTwins2023,
  title = {High-Fidelity Digital Twins: Detecting and Localizing Weaknesses in Structures},
  author = {L{\"o}hner, Rainald and Airaudo, Facundo N. and Antil, Harbir and W{\"u}chner, Roland and Meister, Fabian and Warnakulasuriya, Suneth},
  year = {2023},
  publisher = {arXiv},
  doi = {10.48550/arXiv.2311.10925}
}

@article{youDistributedBendingStiffness2023,
  title = {Distributed Bending Stiffness Estimation of Bridges Using Adaptive Inverse Unit Load Method},
  author = {You, Run-Zhou and Yi, Ting-Hua and Ren, Liang and Li, Hong-Nan},
  year = {2023},
  journal = {Engineering Structures},
  volume = {297},
  pages = {116981},
  doi = {10.1016/j.engstruct.2023.116981}
}

@article{diazMergingExperimentalDesign2023,
  title = {Merging Experimental Design and Structural Identification around the Concept of Modified Constitutive Relation Error in Low-Frequency Dynamics for Enhanced Structural Monitoring},
  author = {Diaz, M. and Charbonnel, P.-{\'E}. and Chamoin, L.},
  year = {2023},
  journal = {Mechanical Systems and Signal Processing},
  volume = {197},
  pages = {110371},
  doi = {10.1016/j.ymssp.2023.110371}
}

@misc{antilFisherInformationBasedSensorPlacement2026,
  title = {Fisher-Information-Based Sensor Placement for Structural Digital Twins: Analytic Results and Benchmarks},
  author = {Antil, Harbir and Jain, Animesh and L{\"o}hner, Rainald},
  year = {2026},
  publisher = {arXiv},
  eprint = {2602.02981},
  archiveprefix = {arXiv},
  doi = {10.48550/arXiv.2602.02981}
}

@misc{bakeerSensorInformativenessIdentifiability2025a,
  title = {Sensor Informativeness, Identifiability, and Uncertainty in Bayesian Inverse Problems for Structural Health Monitoring},
  author = {Bakeer, Tammam and Herbers, Max and Marx, Steffen},
  year = {2025},
  publisher = {arXiv},
  eprint = {2511.16628},
  archiveprefix = {arXiv},
  doi = {10.48550/arXiv.2511.16628}
}

@article{moorePrincipalComponentAnalysis1981,
  title = {Principal Component Analysis in Linear Systems: Controllability, Observability, and Model Reduction},
  author = {Moore, Bruce},
  year = {1981},
  journal = {IEEE Transactions on Automatic Control},
  volume = {26},
  number = {1},
  pages = {17--32},
  doi = {10.1109/TAC.1981.1102568}
}

@article{himpeEmgrEmpiricalGramian2018,
  title = {emgr - The Empirical Gramian Framework},
  author = {Himpe, Christian},
  year = {2018},
  journal = {Algorithms},
  volume = {11},
  number = {7},
  pages = {91},
  doi = {10.3390/a11070091}
}

@article{lewisRoleObservabilityGramian2022,
  title = {Role of the Observability Gramian in Parameter Estimation: Application to Nonchaotic and Chaotic Systems via the Forward Sensitivity Method},
  author = {Lewis, John M. and Lakshmivarahan, Sivaramakrishnan},
  year = {2022},
  journal = {Atmosphere},
  volume = {13},
  number = {10},
  doi = {10.3390/atmos13101647}
}

@article{stuartInverseProblemsBayesian2010,
  title = {Inverse Problems: A Bayesian Perspective},
  author = {Stuart, A. M.},
  year = {2010},
  journal = {Acta Numerica},
  volume = {19},
  pages = {451--559},
  doi = {10.1017/S0962492910000061}
}

@article{flathFastAlgorithmsBayesian2011,
  title = {Fast Algorithms for Bayesian Uncertainty Quantification in Large-Scale Linear Inverse Problems Based on Low-Rank Partial Hessian Approximations},
  author = {Flath, H. P. and Wilcox, L. C. and Ak{\c{c}}elik, V. and Hill, J. and van Bloemen Waanders, B. and Ghattas, O.},
  year = {2011},
  journal = {SIAM Journal on Scientific Computing},
  volume = {33},
  number = {1},
  pages = {407--432},
  doi = {10.1137/090780717}
}

@article{bui-thanhComputationalFrameworkInfiniteDimensional2013,
  title = {A Computational Framework for Infinite-Dimensional Bayesian Inverse Problems Part I: The Linearized Case, with Application to Global Seismic Inversion},
  author = {Bui-Thanh, Tan and Ghattas, Omar and Martin, James and Stadler, Georg},
  year = {2013},
  journal = {SIAM Journal on Scientific Computing},
  volume = {35},
  number = {6},
  pages = {A2494--A2523},
  doi = {10.1137/12089586X}
}

@article{cuiLikelihoodinformedDimensionReduction2014,
  title = {Likelihood-Informed Dimension Reduction for Nonlinear Inverse Problems},
  author = {Cui, Tiangang and Martin, James and Marzouk, Youssef M. and Solonen, Antti and Spantini, Alessio},
  year = {2014},
  journal = {Inverse Problems},
  volume = {30},
  number = {11},
  pages = {114015},
  doi = {10.1088/0266-5611/30/11/114015}
}

@article{spantiniOptimalLowrankApproximations2015,
  title = {Optimal Low-rank Approximations of Bayesian Linear Inverse Problems},
  author = {Spantini, Alessio and Solonen, Antti and Cui, Tiangang and Martin, James and Tenorio, Luis and Marzouk, Youssef},
  year = {2015},
  journal = {SIAM Journal on Scientific Computing},
  volume = {37},
  number = {6},
  pages = {A2451--A2487},
  doi = {10.1137/140977308}
}

@article{alexanderianAOptimalDesignExperiments2014,
  title = {A-Optimal Design of Experiments for Infinite-Dimensional Bayesian Linear Inverse Problems with Regularized {$\ell_0$}-Sparsification},
  author = {Alexanderian, Alen and Petra, Noemi and Stadler, Georg and Ghattas, Omar},
  year = {2014},
  journal = {SIAM Journal on Scientific Computing},
  volume = {36},
  number = {5},
  pages = {A2122--A2148},
  doi = {10.1137/130933381}
}

@article{alexanderianOptimalExperimentalDesign2021,
  title = {Optimal Experimental Design for Infinite-Dimensional Bayesian Inverse Problems Governed by PDEs: A Review},
  author = {Alexanderian, Alen},
  year = {2021},
  journal = {Inverse Problems},
  volume = {37},
  number = {4},
  pages = {043001},
  doi = {10.1088/1361-6420/abe10c}
}

@inproceedings{braceSensorPlacementCantilever2022,
  title = {Sensor Placement on a Cantilever Beam Using Observability Gramians},
  author = {Brace, Natalie L. and Andrews, Nicholas B. and Upsal, Jeremy and Morgansen, Kristi A.},
  year = {2022},
  booktitle = {2022 IEEE 61st Conference on Decision and Control (CDC)},
  pages = {388--395},
  publisher = {IEEE},
  doi = {10.1109/CDC51059.2022.9992639}
}

@article{wielandStructuralPracticalIdentifiability2021,
  title = {On Structural and Practical Identifiability},
  author = {Wieland, Franz-Georg and Hauber, Adrian L. and Rosenblatt, Marcus and T{\"o}nsing, Christian and Timmer, Jens},
  year = {2021},
  journal = {Current Opinion in Systems Biology},
  volume = {25},
  pages = {60--69},
  doi = {10.1016/j.coisb.2021.03.005}
}

@article{bangerthEstimatingUsingInformation2026,
  title = {Estimating and Using Information in Inverse Problems},
  author = {Bangerth, Wolfgang and Johnson, Chris R. and Njeru, Dennis K. and {van Bloemen Waanders}, Bart},
  year = 2026,
  journal = {Inverse Problems and Imaging},
  volume = {24},
  pages = {1--33},
  issn = {1930-8337, 1930-8345},
  doi = {10.3934/ipi.2026003},
  urldate = {2026-05-18},
}

@article{kammerSensorPlacementOnorbit1991,
  title = {Sensor Placement for On-Orbit Modal Identification and Correlation of Large Space Structures},
  author = {Kammer, Daniel C.},
  year = 1991,
  month = mar,
  journal = {Journal of Guidance, Control, and Dynamics},
  volume = {14},
  number = {2},
  pages = {251--259},
  issn = {0731-5090, 1533-3884},
  doi = {10.2514/3.20635},
  urldate = {2026-02-08},
  langid = {english}
}

@article{bertolaOptimalMultiTypeSensor2017,
  title = {Optimal {{Multi-Type Sensor Placement}} for {{Structural Identification}} by {{Static-Load Testing}}},
  author = {Bertola, Numa Joy and Papadopoulou, Maria and Vernay, Didier and Smith, Ian F. C.},
  year = 2017,
  month = dec,
  journal = {Sensors (Basel, Switzerland)},
  volume = {17},
  number = {12},
  pages = {2904},
  issn = {1424-8220},
  doi = {10.3390/s17122904},
  urldate = {2025-08-01},
  langid = {english},
  pmcid = {PMC5751592},
  pmid = {29240684}
}

@misc{powelEmpiricalObservabilityGramian2020,
  title = {Empirical {{Observability Gramian}} for {{Stochastic Observability}} of {{Nonlinear Systems}}},
  author = {Powel, Nathan and Morgansen, Kristi A.},
  year = 2020,
  month = jun,
  number = {arXiv:2006.07451},
  eprint = {2006.07451},
  primaryclass = {eess},
  publisher = {arXiv},
  doi = {10.48550/arXiv.2006.07451},
  urldate = {2026-02-06},
  archiveprefix = {arXiv}
}

@article{kennedyBayesianCalibrationComputer2001,
  title = {Bayesian Calibration of Computer Models},
  author = {Kennedy, Marc C. and O'Hagan, Anthony},
  year = {2001},
  journal = {Journal of the Royal Statistical Society: Series B (Statistical Methodology)},
  volume = {63},
  number = {3},
  pages = {425--464},
  issn = {1369-7412},
  doi = {10.1111/1467-9868.00294},
  langid = {english}
}

@article{alexanderianFastScalableMethod2016,
  title = {A {{Fast}} and {{Scalable Method}} for {{A-Optimal Design}} of {{Experiments}} for {{Infinite-dimensional Bayesian Nonlinear Inverse Problems}}},
  author = {Alexanderian, Alen and Petra, Noemi and Stadler, Georg and Ghattas, Omar},
  year = 2016,
  month = jan,
  journal = {SIAM Journal on Scientific Computing},
  volume = {38},
  number = {1},
  pages = {A243-A272},
  publisher = {{Society for Industrial and Applied Mathematics}},
  issn = {1064-8275},
  doi = {10.1137/140992564},
  urldate = {2026-05-01},
}

@article{boyaciogluDualityStochasticObservability2024,
  title = {Duality of Stochastic Observability and Constructability and Links to Fisher Information},
  author = {Boyac{\i}o{\u g}lu, Burak and {van Breugel}, Floris},
  year = 2024,
  journal = {IEEE Control Systems Letters},
  volume = {8},
  pages = {3458--3463},
  issn = {2475-1456},
  doi = {10.1109/LCSYS.2025.3547297},
  urldate = {2026-05-18},
  copyright = {https://ieeexplore.ieee.org/Xplorehelp/downloads/license-information/IEEE.html}
}

@article{brynjarsdottirLearningPhysicalParameters2014,
  title = {Learning about Physical Parameters: The Importance of Model Discrepancy},
  shorttitle = {Learning about Physical Parameters},
  author = {Brynjarsd{\'o}ttir, Jenn{\'y} and O'Hagan, Anthony},
  year = 2014,
  month = oct,
  journal = {Inverse Problems},
  volume = {30},
  number = {11},
  pages = {114007},
  publisher = {IOP Publishing},
  issn = {0266-5611},
  doi = {10.1088/0266-5611/30/11/114007},
  urldate = {2022-10-01},
  langid = {english},
}

@article{cuiPriorNormalizationCertified2022a,
  title = {Prior Normalization for Certified Likelihood-Informed Subspace Detection of {{Bayesian}} Inverse Problems},
  author = {Cui, Tiangang and Tong, Xin T. and Zahm, Olivier},
  year = 2022,
  month = oct,
  journal = {Inverse Problems},
  volume = {38},
  number = {12},
  pages = {124002},
  publisher = {IOP Publishing},
  issn = {0266-5611},
  doi = {10.1088/1361-6420/ac9582},
  urldate = {2026-05-01},
  langid = {english}
}

@incollection{dashtiBayesianApproachInverse2017a,
  title = {The {{Bayesian Approach}} to {{Inverse Problems}}},
  booktitle = {Handbook of {{Uncertainty Quantification}}},
  author = {Dashti, Masoumeh and Stuart, Andrew M.},
  year = 2017,
  pages = {311--428},
  publisher = {Springer, Cham},
  address = {Cham},
  doi = {10.1007/978-3-319-12385-1_7},
  urldate = {2026-05-01},
  isbn = {978-3-319-12385-1},
  langid = {english}
}

@article{huberRobustEstimationLocation1964a,
  title = {Robust {{Estimation}} of a {{Location Parameter}}},
  author = {Huber, Peter J.},
  year = 1964,
  month = mar,
  journal = {The Annals of Mathematical Statistics},
  volume = {35},
  number = {1},
  pages = {73--101},
  publisher = {Institute of Mathematical Statistics},
  issn = {0003-4851, 2168-8990},
  doi = {10.1214/aoms/1177703732},
  urldate = {2026-05-01},
  langid = {english},
}

@book{huberRobustStatistics2009,
  title = {Robust Statistics},
  author = {Huber, Peter J. and Ronchetti, Elvezio},
  year = 2009,
  series = {Wiley Series in Probability and Statistics},
  edition = {2nd ed},
  publisher = {Wiley},
  address = {Hoboken, N.J},
  doi = {10.1002/9780470434697},
  isbn = {978-0-470-12990-6},
  langid = {english},
  lccn = {QA276 .H785 2009},
}

@article{jauffretObservabilityFisherInformation2007a,
  title = {Observability and Fisher Information Matrix in Nonlinear Regression},
  author = {Jauffret, Claude},
  year = 2007,
  month = apr,
  journal = {IEEE Transactions on Aerospace and Electronic Systems},
  volume = {43},
  number = {2},
  pages = {756--759},
  issn = {1557-9603},
  doi = {10.1109/TAES.2007.4285368},
  urldate = {2026-05-01},
}

@article{lingSelectionModelDiscrepancy2014a,
  title = {Selection of Model Discrepancy Priors in {{Bayesian}} Calibration},
  author = {Ling, You and Mullins, Joshua and Mahadevan, Sankaran},
  year = 2014,
  month = nov,
  journal = {Journal of Computational Physics},
  volume = {276},
  pages = {665--680},
  issn = {0021-9991},
  doi = {10.1016/j.jcp.2014.08.005},
  urldate = {2026-05-01},
}

@article{petraComputationalFrameworkInfiniteDimensional2014,
  title = {A {{Computational Framework}} for {{Infinite-Dimensional Bayesian Inverse Problems}}, {{Part II}}: {{Stochastic Newton MCMC}} with {{Application}} to {{Ice Sheet Flow Inverse Problems}}},
  shorttitle = {A {{Computational Framework}} for {{Infinite-Dimensional Bayesian Inverse Problems}}, {{Part II}}},
  author = {Petra, Noemi and Martin, James and Stadler, Georg and Ghattas, Omar},
  year = 2014,
  month = jan,
  journal = {SIAM Journal on Scientific Computing},
  volume = {36},
  number = {4},
  pages = {A1525-A1555},
  publisher = {{Society for Industrial and Applied Mathematics}},
  issn = {1064-8275},
  doi = {10.1137/130934805},
  urldate = {2026-05-01},
}

@article{qianModelReductionLinear2022,
  title = {Model {{Reduction}} of {{Linear Dynamical Systems}} via {{Balancing}} for {{Bayesian Inference}}},
  author = {Qian, Elizabeth and Tabeart, Jemima M. and Beattie, Christopher and Gugercin, Serkan and Jiang, Jiahua and Kramer, Peter R. and Narayan, Akil},
  year = 2022,
  month = mar,
  journal = {Journal of Scientific Computing},
  volume = {91},
  number = {1},
  pages = {29},
  issn = {1573-7691},
  doi = {10.1007/s10915-022-01798-8},
  urldate = {2026-05-01},
  langid = {english}
}

@article{ryanReviewModernComputational2016,
  title = {A Review of Modern Computational Algorithms for Bayesian Optimal Design},
  author = {Ryan, Elizabeth G. and Drovandi, Christopher C. and McGree, James M. and Pettitt, Anthony N.},
  year = 2016,
  journal = {International Statistical Review},
  volume = {84},
  number = {1},
  pages = {128--154},
  doi = {10.1111/insr.12107}
}

@article{bagheriIdentificationFlexuralRigidity2018,
  title = {Identification of Flexural Rigidity in Bridges with Limited Structural Information},
  author = {Bagheri, Abdollah and Alipour, Mohamad and Ozbulut, Osman E. and Harris, Devin K.},
  year = {2018},
  journal = {Journal of Structural Engineering},
  volume = {144},
  number = {8},
  pages = {04018126},
  doi = {10.1061/(ASCE)ST.1943-541X.0002131}
}

@article{leeFrequencydomainMethodStructural2002,
  title = {A Frequency-Domain Method of Structural Damage Identification Formulated from the Dynamic Stiffness Equation of Motion},
  author = {Lee, U. and Shin, J.},
  year = {2002},
  journal = {Journal of Sound and Vibration},
  volume = {257},
  number = {4},
  pages = {615--634},
  doi = {10.1006/jsvi.2002.5058}
}

@article{panResonancebasedApproachSection2021,
  title = {Resonance-Based Approach for Section Flexural Rigidity Identification of Simply Supported Beams},
  author = {Pan, Danguang and Feng, Zhiyao and Lu, Pan and Zheng, Zijian and Zhao, Bincheng},
  year = {2021},
  journal = {Engineering Structures},
  volume = {236},
  pages = {112070},
  doi = {10.1016/j.engstruct.2021.112070}
}

@article{adhikariDistributedParameterModel2010,
  title = {Distributed Parameter Model Updating Using the {{Karhunen}}--{{Lo\`eve}} Expansion},
  author = {Adhikari, S. and Friswell, M. I.},
  year = {2010},
  journal = {Mechanical Systems and Signal Processing},
  volume = {24},
  number = {2},
  pages = {326--339},
  doi = {10.1016/j.ymssp.2009.08.007}
}

@article{bonkowskiStiffnessIdentificationReinforced2024,
  title = {Stiffness Identification of Reinforced Concrete Beams Using Rotation Rate Sensors},
  author = {Bo{\'n}kowski, Piotr Adam and Bobra, Piotr and Zembaty, Zbigniew and J{\c e}draszak, Bronis{\l}aw},
  year = {2024},
  journal = {Engineering Structures},
  volume = {307},
  pages = {117969},
  doi = {10.1016/j.engstruct.2024.117969}
}

@article{kooSubstructuralIdentificationFlexural2015,
  title = {Substructural Identification of Flexural Rigidity for Beam-Like Structures},
  author = {Koo, Ki-Young and Yi, Jin-Hak},
  year = {2015},
  journal = {Shock and Vibration},
  volume = {2015},
  pages = {1--15},
  doi = {10.1155/2015/726410}
}

@article{attiaGoalorientedOptimalDesign2018,
  title = {Goal-Oriented Optimal Design of Experiments for Large-Scale {{Bayesian}} Linear Inverse Problems},
  author = {Attia, Ahmed and Alexanderian, Alen and Saibaba, Arvind K.},
  year = 2018,
  month = jul,
  journal = {Inverse Problems},
  volume = {34},
  number = {9},
  publisher = {IOPscience},
  issn = {0266-5611},
  doi = {10.1088/1361-6420/aad210},
  urldate = {2026-05-22},
  langid = {english},
}

@article{attiaOptimalExperimentalDesign2022,
  title = {Optimal {{Experimental Design}} for {{Inverse Problems}} in the {{Presence}} of {{Observation Correlations}}},
  author = {Attia, Ahmed and Constantinescu, Emil},
  year = 2022,
  month = aug,
  journal = {SIAM Journal on Scientific Computing},
  volume = {44},
  publisher = {Argonne National Laboratory (ANL)},
  issn = {1064-8275},
  doi = {10.1137/21M1418666},
  urldate = {2026-05-22},
  langid = {english}
}

@article{barthorpeEmergingTrendsOptimal2020,
  title = {Emerging {{Trends}} in {{Optimal Structural Health Monitoring System Design}}: {{From Sensor Placement}} to {{System Evaluation}}},
  shorttitle = {Emerging {{Trends}} in {{Optimal Structural Health Monitoring System Design}}},
  author = {Barthorpe, Robert James and Worden, Keith},
  year = 2020,
  month = sep,
  journal = {Journal of Sensor and Actuator Networks},
  volume = {9},
  number = {3},
  pages = {31},
  publisher = {Multidisciplinary Digital Publishing Institute},
  issn = {2224-2708},
  doi = {10.3390/jsan9030031},
  urldate = {2026-05-22},
  copyright = {http://creativecommons.org/licenses/by/3.0/},
  langid = {english},
}

@article{bellMultiresponseParameterEstimation2007,
  title = {Multiresponse {{Parameter Estimation}} for {{Finite-Element Model Updating Using Nondestructive Test Data}}},
  author = {Bell, Erin Santini and Sanayei, Masoud and Javdekar, Chitra N. and Slavsky, Eugene},
  year = 2007,
  month = aug,
  journal = {Journal of Structural Engineering},
  volume = {133},
  number = {8},
  pages = {1067--1079},
  publisher = {American Society of Civil Engineers},
  issn = {0733-9445},
  doi = {10.1061/(ASCE)0733-9445(2007)133:8(1067)},
  urldate = {2026-05-22},
  langid = {english},
}

@article{ebrahimianInformationTheoreticApproachIdentifiability2019,
  title = {Information-{{Theoretic Approach}} for {{Identifiability Assessment}} of {{Nonlinear Structural Finite-Element Models}}},
  author = {Ebrahimian, Hamed and Astroza, Rodrigo and Conte, Joel P. and Bitmead, Robert R.},
  year = 2019,
  month = jul,
  journal = {Journal of Engineering Mechanics},
  volume = {145},
  number = {7},
  pages = {04019039},
  publisher = {American Society of Civil Engineers},
  issn = {0733-9399},
  doi = {10.1061/(ASCE)EM.1943-7889.0001590},
  urldate = {2026-05-22},
  langid = {english},
}

@article{ercanBayesianOptimalSensor2023,
  title = {Bayesian Optimal Sensor Placement for Parameter Estimation under Modeling and Input Uncertainties},
  author = {Ercan, Tulay and Papadimitriou, Costas},
  year = 2023,
  month = oct,
  journal = {Journal of Sound and Vibration},
  volume = {563},
  pages = {117844},
  issn = {0022-460X},
  doi = {10.1016/j.jsv.2023.117844},
  urldate = {2026-05-22},
}

@article{huanOptimalExperimentalDesign2024,
  title = {Optimal Experimental Design: {{Formulations}} and Computations},
  shorttitle = {Optimal Experimental Design},
  author = {Huan, Xun and Jagalur, Jayanth and Marzouk, Youssef},
  year = 2024,
  month = jul,
  journal = {Acta Numerica},
  volume = {33},
  pages = {715--840},
  issn = {0962-4929, 1474-0508},
  doi = {10.1017/S0962492924000023},
  urldate = {2026-05-22},
  langid = {english},
}

@article{huanSimulationbasedOptimalBayesian2013,
  title = {Simulation-Based Optimal {{Bayesian}} Experimental Design for Nonlinear Systems},
  author = {Huan, Xun and Marzouk, Youssef M.},
  year = 2013,
  month = jan,
  journal = {Journal of Computational Physics},
  volume = {232},
  number = {1},
  pages = {288--317},
  issn = {00219991},
  doi = {10.1016/j.jcp.2012.08.013},
  urldate = {2026-05-22},
  copyright = {https://www.elsevier.com/tdm/userlicense/1.0/},
  langid = {english},
}

@article{kimSequentialFrameworkImproving2019,
  title = {A {{Sequential Framework}} for {{Improving Identifiability}} of {{FE Model Updating Using Static}} and {{Dynamic Data}}},
  author = {Kim, Sehoon and Kim, Namgyu and Park, Young-Soo and Jin, Seung-Seop},
  year = 2019,
  month = jan,
  journal = {Sensors},
  volume = {19},
  number = {23},
  pages = {5099},
  publisher = {Multidisciplinary Digital Publishing Institute},
  issn = {1424-8220},
  doi = {10.3390/s19235099},
  urldate = {2026-05-22},
  copyright = {http://creativecommons.org/licenses/by/3.0/},
  langid = {english},
}

@article{kunwooObservabilityGramianBayesian2023,
  title = {Observability {{Gramian}} for {{Bayesian Inference}} in {{Nonlinear Systems With Its Industrial Application}}},
  author = {Kunwoo, Lee and Umezu, Yusuke and Konno, Kaiki and Kashima, Kenji},
  year = 2023,
  journal = {IEEE Control Systems Letters},
  volume = {7},
  pages = {871--876},
  issn = {2475-1456},
  doi = {10.1109/LCSYS.2022.3227452},
  urldate = {2026-05-22},
}

@article{lallEmpiricalModelReduction1999,
  title = {Empirical Model Reduction of Controlled Nonlinear Systems},
  author = {Lall, Sanjay and Marsden, Jerrold E. and Glava{\v s}ki, Sonja},
  year = 1999,
  month = jul,
  journal = {IFAC Proceedings Volumes},
  series = {14th {{IFAC World Congress}} 1999, {{Beijing}}, {{Chia}}, 5-9 {{July}}},
  volume = {32},
  number = {2},
  pages = {2598--2603},
  publisher = {Elsevier},
  issn = {1474-6670},
  doi = {10.1016/S1474-6670(17)56442-3},
  urldate = {2026-04-10},
}

@misc{menkeResolutionCovarianceGLS2014,
  title = {Resolution and Covariance in Generalized Least Squares Inversion},
  author = {Menke, William},
  year = 2014,
  urldate = {2026-05-22},
}

@inproceedings{sanayeiSensorPlacementParameter2002,
  title = {Sensor {{Placement}} for {{Parameter Estimation}} of {{Structures Using Fisher Information Matrix}}},
  booktitle = {Applications of {{Advanced Technologies}} in {{Transportation}} (2002)},
  author = {Sanayei, Masoud and Javdekar, Chitra N.},
  year = 2002,
  month = jul,
  pages = {385--386},
  publisher = {American Society of Civil Engineers},
  address = {Boston Marriot, Cambridge, Massachusetts, United States},
  doi = {10.1061/40632(245)49},
  urldate = {2026-05-22},
  isbn = {978-0-7844-0632-8},
  langid = {english},
}

@article{scherpenBalancingNonlinearSystems1993,
  title = {Balancing for Nonlinear Systems},
  author = {Scherpen, J. M. A.},
  year = 1993,
  month = aug,
  journal = {Systems \& Control Letters},
  volume = {21},
  number = {2},
  pages = {143--153},
  issn = {0167-6911},
  doi = {10.1016/0167-6911(93)90117-O},
  urldate = {2026-05-22},
}

@article{vandewouwerApproachSelectionOptimal2000,
  title = {An Approach to the Selection of Optimal Sensor Locations in Distributed Parameter Systems},
  author = {Vande Wouwer, Alain and Point, Nicolas and Porteman, St{\'e}phanie and Remy, Marcel},
  year = 2000,
  month = aug,
  journal = {Journal of Process Control},
  volume = {10},
  number = {4},
  pages = {291--300},
  issn = {0959-1524},
  doi = {10.1016/S0959-1524(99)00048-7},
  urldate = {2026-05-22},
}

@article{yiMethodologyDevelopmentsSensor2012,
  title = {Methodology {{Developments}} in {{Sensor Placement}} for {{Health Monitoring}} of {{Civil Infrastructures}}},
  author = {Yi, Ting-Hua and Li, Hong-Nan},
  year = 2012,
  month = aug,
  journal = {International Journal of Distributed Sensor Networks},
  volume = {8},
  number = {8},
  pages = {612726},
  publisher = {SAGE Publications},
  issn = {1550-1329},
  doi = {10.1155/2012/612726},
  urldate = {2026-05-22},
  langid = {english},
}

@article{ashyraliyevSystemsBiologyParameter2009,
  title = {Systems Biology: Parameter Estimation for Biochemical Models},
  shorttitle = {Systems Biology},
  author = {Ashyraliyev, Maksat and {Fomekong-Nanfack}, Yves and Kaandorp, Jaap A. and Blom, Joke G.},
  year = 2009,
  journal = {The FEBS Journal},
  volume = {276},
  number = {4},
  pages = {886--902},
  issn = {1742-4658},
  doi = {10.1111/j.1742-4658.2008.06844.x},
  urldate = {2026-05-23},
  copyright = {\copyright{} 2009 The Authors Journal compilation \copyright{} 2009 FEBS},
  langid = {english},
}

@article{qiOptimalExperimentalDesign2026,
  title = {Optimal Experimental Design for Parameter Estimation in the Presence of Observation Noise},
  author = {Qi, Jie and Baker, Ruth E.},
  year = 2026,
  month = feb,
  journal = {Mathematical Biosciences},
  volume = {392},
  pages = {109571},
  issn = {0025-5564},
  doi = {10.1016/j.mbs.2025.109571},
  urldate = {2026-05-23},
}

@article{airaudoAdjointbasedDeterminationWeaknesses2023,
  title = {Adjoint-Based Determination of Weaknesses in Structures},
  author = {Airaudo, Facundo N. and L{\"o}hner, Rainald and W{\"u}chner, Roland and Antil, Harbir},
  year = 2023,
  month = dec,
  journal = {Computer Methods in Applied Mechanics and Engineering},
  volume = {417},
  pages = {116471},
  issn = {00457825},
  doi = {10.1016/j.cma.2023.116471},
  urldate = {2026-05-24},
  langid = {english},
}

@misc{scheffelsLikelihoodinformedModelReduction2025,
  title = {Likelihood-Informed {{Model Reduction}} for {{Bayesian Inference}} of {{Static Structural Loads}}},
  author = {Scheffels, Jakob and Qian, Elizabeth and Papaioannou, Iason and Ullmann, Elisabeth},
  year = 2025,
  month = oct,
  number = {arXiv:2510.07950},
  eprint = {2510.07950},
  primaryclass = {math.NA},
  publisher = {arXiv},
  doi = {10.48550/arXiv.2510.07950},
  urldate = {2026-05-24},
  archiveprefix = {arXiv},
}
\end{document}